\documentclass[copyright]{article}
 
\usepackage[margin=1in]{geometry}

\usepackage{geometry} 

\usepackage{graphicx} 

\renewcommand{\Omega}{{\bf \emptyset}}

\usepackage[style=exampledefault,roundcorner=20,linecolor=black,middlelinewidth=2pt,leftmargin=1cm,rightmargin=1cm]{mdframed}

\usepackage{stmaryrd}
\usepackage{extarrows}
\usepackage{mathtools}
\usepackage{float}
\usepackage{hyperref}
\usepackage{fixltx2e}
\usepackage{multicol}
\usepackage{amsfonts}
\usepackage{amsmath}
\usepackage{amsthm}
\usepackage{enumitem}
\usepackage{tikz}
\usepackage{float}
\usepackage[all]{xy}
\usepackage{color,soul}
\usepackage{afterpage}
\usepackage{longtable}
\usepackage{sansmath} 

\newdir{ >}{{}*!/-8pt/@{>}}

\theoremstyle{definition}
\newtheorem{theorem}{Theorem}[section]

\newtheorem{definition}[theorem]{Definition}
\newtheorem{lemma}[theorem]{Lemma}
\newtheorem{proposition}[theorem]{Proposition}

\renewcommand{\bar}[1]{\overline{#1}\hspace*{.02cm}}

\newcommand{\TOF}{\mathsf{TOF}}
\newcommand{\CNOT}{\mathsf{CNOT}}
\newcommand{\cnot}{\mathsf{cnot}}
\newcommand{\tof}{\mathsf{tof}}
\newcommand{\notgate}{\mathsf{not}}

\newcommand{\N}{\mathbb{N}}

\newcommand{\C}{\mathbb{C}}
\newcommand{\Z}{\mathbb{Z}}

\newcommand{\FdHilb}{\FHilb}
\newcommand{\FHilb}{\mathsf{FHilb}}
\newcommand{\Mat}{\mathsf{Mat}}
\newcommand{\Not}{\mathsf{not}}

\newcommand{\ZX}{\mathsf{ZX}}

\newcommand{\U}{\mathsf{U}}

\newcommand{\<}{\langle}

\renewcommand{\>}{\rangle}

\newcommand{\onein}{|1\>}
\newcommand{\oneout}{\<1|}

\newcommand{\zeroin}{|0\>}
\newcommand{\zeroout}{\<0|}

\usepackage{relsize}

\let\oldcirc\circ

\renewcommand{\circ}{{\mathrel{\mathsmaller{\mathsmaller{\oldcirc}}}}}

\newcommand{\xpi}{
\begin{tikzpicture}
	\begin{pgfonlayer}{nodelayer}
		\node [style=xmul] (0) at (4, -0) {$\pi$};
		\node [style=rn] (1) at (3, -0) {};
		\node [style=rn] (2) at (5, -0) {};
	\end{pgfonlayer}
	\begin{pgfonlayer}{edgelayer}
		\draw [style=simple] (2) to (0);
		\draw [style=simple] (0) to (1);
	\end{pgfonlayer}
\end{tikzpicture}
}

\newcommand{\zpi}{
\begin{tikzpicture}
	\begin{pgfonlayer}{nodelayer}
		\node [style=zmul] (0) at (4, -0) {$\pi$};
		\node [style=rn] (1) at (3, -0) {};
		\node [style=rn] (2) at (5, -0) {};
	\end{pgfonlayer}
	\begin{pgfonlayer}{edgelayer}
		\draw [style=simple] (2) to (0);
		\draw [style=simple] (0) to (1);
	\end{pgfonlayer}
\end{tikzpicture}
}

\newcommand{\zmul}{
\begin{tikzpicture}
	\begin{pgfonlayer}{nodelayer}
		\node [style=zmul] (0) at (0, -0) {};
		\node [style=rn] (1) at (0.75, -0) {};
		\node [style=rn] (2) at (-0.75, 0.25) {};
		\node [style=rn] (3) at (-0.75, -0.25) {};
	\end{pgfonlayer}
	\begin{pgfonlayer}{edgelayer}
		\draw [style=simple] (0) to (1);
		\draw [style=simple, in=-165, out=0, looseness=1.00] (3) to (0);
		\draw [style=simple, in=0, out=165, looseness=1.00] (0) to (2);
	\end{pgfonlayer}
\end{tikzpicture}
}

\newcommand{\zcomulu}{
\begin{tikzpicture}[xscale=-1]
	\begin{pgfonlayer}{nodelayer}
		\node [style=rn] (0) at (0, -0) {};
		\node [style=zmul] (1) at (-0.75, -0) {};
	\end{pgfonlayer}
	\begin{pgfonlayer}{edgelayer}
		\draw [style=simple] (0) to (1);
	\end{pgfonlayer}
\end{tikzpicture}
}

\newcommand{\zcomul}{
\begin{tikzpicture}[xscale=-1]
	\begin{pgfonlayer}{nodelayer}
		\node [style=zmul] (0) at (0, -0) {};
		\node [style=rn] (1) at (0.75, -0) {};
		\node [style=rn] (2) at (-0.75, 0.25) {};
		\node [style=rn] (3) at (-0.75, -0.25) {};
	\end{pgfonlayer}
	\begin{pgfonlayer}{edgelayer}
		\draw [style=simple] (0) to (1);
		\draw [style=simple, in=-165, out=0, looseness=1.00] (3) to (0);
		\draw [style=simple, in=0, out=165, looseness=1.00] (0) to (2);
	\end{pgfonlayer}
\end{tikzpicture}
}

\newcommand{\zmulu}{
\begin{tikzpicture}
	\begin{pgfonlayer}{nodelayer}
		\node [style=rn] (0) at (0, -0) {};
		\node [style=zmul] (1) at (-0.75, -0) {};
	\end{pgfonlayer}
	\begin{pgfonlayer}{edgelayer}
		\draw [style=simple] (0) to (1);
	\end{pgfonlayer}
\end{tikzpicture}
}

\newcommand{\hadamard}{
\begin{tikzpicture}
	\begin{pgfonlayer}{nodelayer}
		\node [style=rn] (0) at (1, 2.5) {};
		\node [style=rn] (1) at (-1, 2.5) {};
		\node [style=h] (2) at (0, 2.5) {};
	\end{pgfonlayer}
	\begin{pgfonlayer}{edgelayer}
		\draw [style=simple] (0) to (2);
		\draw [style=simple] (2) to (1);
	\end{pgfonlayer}
\end{tikzpicture}
}

\tikzstyle{strings}=[baseline={([yshift=-.5ex]current bounding box.center)}]

\usetikzlibrary{decorations, decorations.text,backgrounds}
\usepackage{xcolor}
\usepackage{graphicx,amsmath,latexsym,amssymb,amsthm,geometry}
\usepackage{pagecolor}

\tikzset{every picture/.append style={scale=.4}, transform shape, strings}

\tikzset{every picture/.append style={background rectangle/.style={fill=black!5}, show background rectangle}}

\renewcommand{\epsilon}{\varepsilon}
\renewcommand{\phi}{\varphi}

\tikzset{%
symbol/.style={%
draw=none,
every to/.append style={%
edge node={node [sloped, allow upside down, auto=false]{$#1$}}}
}
}

\usetikzlibrary{shapes.geometric}
\usetikzlibrary{patterns}
\usetikzlibrary{fit}
\usetikzlibrary{positioning}
\usetikzlibrary{calc}
\usetikzlibrary{arrows}
\usetikzlibrary{decorations.markings}
\usetikzlibrary{decorations.pathreplacing}
\usetikzlibrary{shapes}

\pgfdeclarelayer{background}
\pgfdeclarelayer{nodelayer}
\pgfdeclarelayer{edgelayer}
\pgfsetlayers{background,edgelayer,nodelayer,main}

\tikzset{simple/.style={}}
\tikzset{nothing/.style={outer sep=-3.4pt}}
\tikzset{map/.style={draw,color=black,fill=white, rectangle}}

\tikzset{dot/.style={thick, fill=black, circle, scale=1.3, inner sep = .05cm}}

\tikzset{oplus/.style={draw, scale=0.9,minimum height=.1cm,circle,append after command={
[shorten >=\pgflinewidth, shorten <=\pgflinewidth,]
(\tikzlastnode.north) edge (\tikzlastnode.south)
(\tikzlastnode.east) edge (\tikzlastnode.west)
}
}
}

\tikzset{tri/.style={
draw,
shape border rotate=30,
regular polygon,
regular polygon sides=3,
fill=white,
inner sep = .04cm,
scale=3
}
}

\tikzset{triangle/.style={
draw,
shape border rotate=30,
regular polygon,
regular polygon sides=3,
fill=white,
inner sep = .04cm,
scale=3
}
}

\tikzset{triflip/.style={
draw,
shape border rotate=-30,
regular polygon,
regular polygon sides=3,
fill=white,
inner sep = .04cm,
scale=3
}
}

\tikzset{fanin/.style={
draw,
shape border rotate=30,
regular polygon,
regular polygon sides=3,
fill=white,
inner sep = .1cm
}
}

\tikzset{fanout/.style={
draw,
shape border rotate=-30,
regular polygon,
regular polygon sides=3,
fill=white,
inner sep = .1cm
}
}

\tikzset{onein/.style={
draw,
shape border rotate=30,
regular polygon,
regular polygon sides=3,
fill=black,
inner sep = .04cm,
scale=1.2
}
}

\tikzset{oneout/.style={
draw,
shape border rotate=-30,
regular polygon,
regular polygon sides=3,
fill=black,
inner sep = .04cm,
scale=1.2
}
}

\tikzset{zeroin/.style={
draw,
shape border rotate=30,
regular polygon,
regular polygon sides=3,
fill=white,
inner sep = .04cm,
scale=1.2
}
}

\tikzset{zeroout/.style={
draw,
shape border rotate=-30,
regular polygon,
regular polygon sides=3,
fill=white,
inner sep = .04cm,
scale=1.2
}
}

\tikzset{xmul/.style={draw,fill=white, circle,scale=1.5}}

\tikzset{zmul/.style={draw,fill={rgb:black,1;white,3}, text=black, circle,scale=1.5,thick}}

\tikzset{h/.style={draw,fill=black,  regular polygon, regular polygon sides=4,scale=.8}}

\tikzset{wires/.style={}}

\tikzset{box/.style={inner sep=0pt, thick, draw=black, text height=1.5ex, text depth=.25ex, text centered, minimum height=3em, anchor=center}}

\makeatletter
\newcommand{\oneraggedpage}{\let\mytextbottom\@textbottom
	\let\mytexttop\@texttop
	\raggedbottom
	\afterpage{%
		\global\let\@textbottom\mytextbottom
		\global\let\@texttop\mytexttop}}

\renewcommand{\phi}{\varphi}

\tikzset{rn/.style={outer sep=-3pt}}

\tikzstyle{uptriangle}=[regular polygon, regular polygon sides=3, scale=1, draw]
\tikzstyle{downtriangle}=[regular polygon, regular polygon sides=3, shape border rotate=180, scale=1, draw]
\tikzstyle{none}=[inner sep=-1pt]

\makeatletter
\def\namedlabel#1#2{\begingroup
   \def\@currentlabel{#2}%
   \label{#1}\endgroup
}
\makeatother

\title{Circuit Relations for Real Stabilizers: \\Towards $\TOF+H$}
\author{
	Comfort, Cole
}

\newtheorem{innercustomgeneric}{\customgenericname}
\providecommand{\customgenericname}{}
\newcommand{\newcustomtheorem}[2]{%
  \newenvironment{#1}[1]
  {%
   \renewcommand\customgenericname{#2}%
   \renewcommand\theinnercustomgeneric{##1}%
   \innercustomgeneric
  }
  {\endinnercustomgeneric}
}

\newcustomtheorem{customproposition}{Proposition}
\newcustomtheorem{customtheorem}{Theorem}
\newcustomtheorem{customlemma}{Lemma}

\begin{document}
\maketitle

\allowdisplaybreaks

\begin{abstract}
The real stabilizer fragment of quantum mechanics was shown to have a complete axiomatization in terms of the angle-free fragment of the ZX-calculus. This fragment of the ZX-calculus---although abstractly elegant---is stated in terms of identities, such as spider fusion which generally do not have interpretations as circuit transformations.

We complete the category $\CNOT$ generated by the controlled not gate and the computational ancillary bits, presented by circuit relations, to the real stabilizer fragment of quantum mechanics.  This is performed first, by adding the Hadamard gate and the scalar $\sqrt{2}$ as generators.  We then construct translations to and from the angle-free fragment of the ZX-calculus, showing that they are inverses. 

 We then discuss how this could potentially lead to a complete axiomatization, in terms of circuit relations, for the approximately universal fragment of quantum mechanics generated by the Toffoli gate, Hadamard gate and computational ancillary bits.

\end{abstract}

\section{Background}

The angle-free fragment of the $\ZX$ calculus---describing the interaction of the $Z$ and $X$ observables, Hadamard gate and $\pi$ phases---is known to be complete for (pure) real stabilizer circuits (stabilizer circuits with real coefficients) \cite{realzx}.
In \cite[Section 4.1]{realzx}, it is shown that real stabilizer circuits are generated by the controlled-$Z$ gate, the $Z$ gate, the Hadamard gate, $|0\>$ state preparations and $\<0|$ post-selected measurements.  Therefore, real stabilizer circuits can also be generated by the controlled-not gate, Hadamard gate, $|1\>$ state preparations and $\<1|$ post-selected measurements.
Although the Hadamard gate, controlled-not gate and computational ancillary bits are derivable in the angle-free fragment of the ZX-calculus, the identities are not given in terms of circuit relations involving these gates: and instead, on identities such as spider fusion which do not preserve the causal structure of being a circuit. Therefore circuit simplification usually involves a circuit reconstruction step at the end.

We provide a complete set of {\em circuit identities} for the category generated by the controlled-not gate, Hadamard gate and state preparation for $|1\>$ and postselected measurement for $\<1|$ and the scalar $\sqrt 2$. The completeness is proven by performing a translation to and from the angle-free fragment of the ZX-calculus; however, in contrast to the ZX-calculus, we structure the identities so that they preserve the causal structure of circuits.  Although the axiomatization  we describe is just as computationally expressive as the angle-free fragment of the ZX-calculus, it provides a high-level language for real stabilizer circuits; and hopefuly, alongside the circuit axiomatization for Toffoli circuits, will lead to a high-level language for Hadamard+Toffoli circuits.


The compilation of quantum programming languages can involve multiple intermediate steps before an optimized physical-level circuit is produced; where optimization can be performed at various levels of granularity \cite{maslov2017basic}.
At the coarsest level of granularity, classical oracles and subroutines are synthesized using generalized controlled-not gates.
Next these generalized controlled-not gates are decomposed into cascading Toffoli gates, and so on...
Eventually, these gates are decomposed into fault tolerant 1 and 2-qubit gates.

The ZX-calculus has  proven to be successful for this fine grain optimization, in particular at reducing $T$ counts \cite{190310477}. This optimization is performed in three steps: first, a translation must be performed turning  the circuit into spiders and phases.  Second, the spider fusion laws, Hopf laws, bialgebra laws and so on are applied to reduce the number of nodes/phases; transforming the circuit into a simpler form resembling an undirected, labeled graph without a global causal structure.  Finally, an optimized circuit is re-extracted from this undirected graph.
In  order to extract circuits at the end, for example, \cite{zxsimp,duncan2010rewriting} use a property of graphs called gFlow.

Using only circuit relations, in contrast,  \cite{Fagan}  were able to reduce 2-qubit Clifford circuits to minimal forms in quantomatic.

Toffoli+Hadamard  quantum circuits, as opposed to the ZX-calculus, are more suitably a language for classical oracles, and thus, are appropriate for coarse granularity optimization.  The controlled-not+Hadamard  subfragment, on the other hand, which we discuss in this paper one can only produce oracles for affine Boolean functions---which is obviously very computationally weak.
The eventual goal, however, is to use this complete axiomatization controlled-not+Hadamard circuits given in this paper, and the axiomatization of Toffoli circuits provided in \cite{tof}, to provide a complete set of identities for the approximately universal \cite{tofh} fragment Toffoli+Hadamard circuits.  In this fragment, indeed, all oracles for classical  Boolean functions can be constructed \cite{tof, aaronson}.  
In Section \ref{section:tof}, we discuss how this circuit axiomatization of controlled-not+Hadamard circuits could potentially lead to one for Toffoli+Hadamard circuits,

Toffoli+Hadamard circuits also easily accommodate the notion of quantum control.
This is useful for implementing circuits corresponding to the conditional execution of various subroutines; which is discussed in \cite[Section 2.4.3]{gilesprogramming} and \cite{hanersoftware}.
Although, in the fragment which we discuss in this paper, we can not control all unitaries: namely circuits containing controlled-not gates can not be controlled. Again, the eventual goal is to extend the axiomatizations of cnot+Hadamard and Toffoli  circuits to Toffoli+Hadamard circuits, where there is no such limitation.

In the ZX-calculus, by contrast, this notion of control is highly unnatural. One would likely have to appeal to the triangle gate, as discussed in \cite{ngcompleteness,triangle,vilmart}.

\section{The controlled not gate}
\label{sec:cnot}

Recall that $\CNOT$ is the PROP generated by the 1 ancillary bits $\onein$ and $\oneout$ as well as the controlled not gate:   
  \[ 
   \onein := 

		\]

Where ``gaps'' are drawn between $\cnot$ gates and $\cnot$ gates are drawn upside down to suppress symmetry maps:

$$
%
$$

These gates must satisfy the identities given in Figure \ref{fig:CNOT}:

\begin{figure}[H]
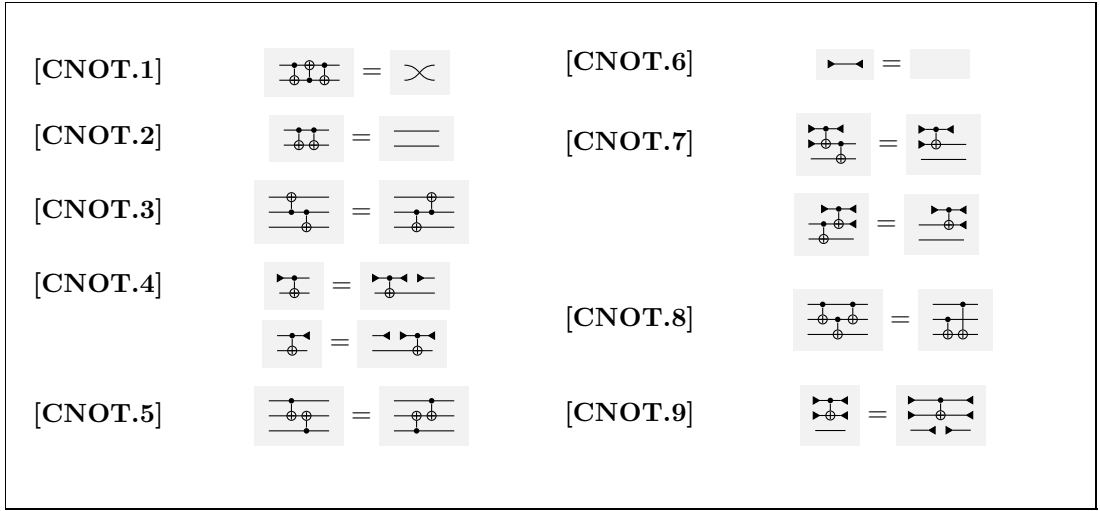

	\noindent
	\scalebox{1.0}{%
		\vbox{%
			\begin{mdframed}
				\begin{multicols}{2}
					\begin{enumerate}[label={\bf [CNOT.\arabic*]}, ref={\bf [CNOT.\arabic*]}, wide = 0pt, leftmargin = 2em]
						\item
						\label{CNOT.1}
						{\hfil
							$
							%
							$}
					\end{enumerate}
				\end{multicols}
				\
			\end{mdframed}
	}}
	\caption{The identities of $\CNOT$}
	\label{fig:CNOT}
\end{figure}

Where the not gate and $|0\>$ ancillary bits are derived:

$$
\Not=
%
$$

Where there is a pseduo-Frobenius structure (non-unital classical structure) generated by:

$$
%
$$

There is the following completeness result:
\begin{theorem}
$\CNOT$ is discrete-inverse equivalent to the category of affine partial isomorphisms between finite-dimensional $\Z_2$ vector spaces, and thus, is complete.
\end{theorem}

\section{Stabilizer quantum mechanics and the angle-free ZX-calculus}

In this section, we briefly describe the well known fragment of quantum mechanics known as stabilizer quantum mechanics.  In particular we focus on the real fragment of stabilizer mechanics, and describe a complete axiomization thereof called the angle-free ZX-calculus.  Stabilizer quantum mechanics are very well studied, a good reference from a categorical perspective is given in \cite{backensthesis}.

\begin{definition} 
The {\bf Pauli matrices} are the complex matrices:
$$
X:=
\begin{bmatrix}
    0 & 1\\
    1 & 0
\end{bmatrix}
\hspace*{.5cm}
Y:=
\begin{bmatrix}
    0 & -i\\
    i & 0
\end{bmatrix}
\hspace*{.5cm}
Z:=
\begin{bmatrix}
    1 & 0\\
    0 & -1
\end{bmatrix}
$$

The {\bf Pauli group} on $n$ is the closure of the set:
$$P_n:=\{\lambda a_1\otimes\cdots\otimes a_n | \lambda \in \{\pm1,\pm i\}, a_i \in \{I_2,X,Y,Z\}\}$$
under matrix multiplication.

The {\bf stabilizer group} of  $|\phi\>$ denoted by $S_{|\phi\>}$ a quantum state is the group of operators for which $|\phi\>$ is a +1 eigenvector.
A state is a  {\bf stabilizer state} in case it is stabilized by a subgroup of $P_n$.
\end{definition}

\begin{definition}
The {\bf Clifford group} on $n$ is the group of operators which acts on the Pauli group on $n$ by conjugation:
$$
C_n:=\{U \in \U(2^n) |\forall p \in P_n, UpU^{-1} \in P_n \}
$$
\end{definition}

There is an algebraic description of stabilizer states:

\begin{lemma}
All $n$ qubit stabilizer states have the form $C|0\>^{\otimes n}$, for some member $C$ of the Clifford group on $n$ qubits.
\end{lemma}

Indeed, we also consider a subgroup of $C_n$:

\begin{definition}
The  {\bf real Clifford group} on $n$ qubits, is the subgroup of the Clifford group with real elements, ie:

$$
C_n^{re}:=\{ U \in C_n | \bar{U} = U\}
$$

So that an $n$-qubit {\bf real stabilizer state } is a state of the form  $C|0\>^{\otimes n}$ for some real Clifford operator $C$.
\end{definition}

We say that a {\bf (real) stabilizer circuit} is a (real) Clifford composed with state preparations and measurements in the computational basis.

\label{sec:ZX}
 The {\bf ZX-calculus} is a collection of calculi describing the interaction of the complementary Frobenius algebras corresponding to the Pauli  $Z$ and $X$ observables and their phases.   The first iteration of the ZX-calculus was described in \cite{coecke2011interacting}.

  \begin{definition}
\label{def:frob}

  A {\bf Frobenius algebra} in a monoidal category is a $5$-tuple:
  $$(A,\zmul,\zmulu,\zcomul,\zcomulu)$$
 such that $(A,\zmul,\zmulu)$ is a monoid and $(A,\zcomul,\zcomulu)$ is a comonoid:

$$
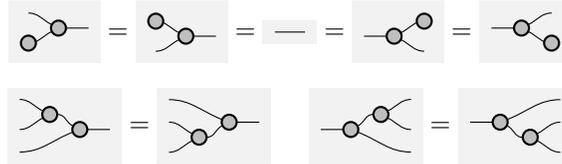
\begin{tikzpicture}[yscale=-1,xscale=-1]
	\begin{pgfonlayer}{nodelayer}
		\node [style=zmul] (0) at (-2, 0.5) {};
		\node [style=none] (1) at (-1, -0) {};
		\node [style=zmul] (2) at (-1, 1) {};
		\node [style=none] (3) at (-3, 0.5) {};
	\end{pgfonlayer}
	\begin{pgfonlayer}{edgelayer}
		\draw [style=simple, in=-27, out=180, looseness=1.00] (1.center) to (0);
		\draw [style=simple, in=180, out=27, looseness=1.00] (0) to (2.center);
		\draw [style=simple] (3.center) to (0);
	\end{pgfonlayer}
\end{tikzpicture}
=
\begin{tikzpicture}[xscale=-1]
	\begin{pgfonlayer}{nodelayer}
		\node [style=zmul] (0) at (-2, 0.5) {};
		\node [style=none] (1) at (-1, -0) {};
		\node [style=zmul] (2) at (-1, 1) {};
		\node [style=none] (3) at (-3, 0.5) {};
	\end{pgfonlayer}
	\begin{pgfonlayer}{edgelayer}
		\draw [style=simple, in=-27, out=180, looseness=1.00] (1.center) to (0);
		\draw [style=simple, in=180, out=27, looseness=1.00] (0) to (2.center);
		\draw [style=simple] (3.center) to (0);
	\end{pgfonlayer}
\end{tikzpicture}
=
\begin{tikzpicture}
	\begin{pgfonlayer}{nodelayer}
		\node [style=none] (0) at (-2, 0.5) {};
		\node [style=none] (1) at (-3, 0.5) {};
	\end{pgfonlayer}
	\begin{pgfonlayer}{edgelayer}
		\draw [style=simple] (0.center) to (1.center);
	\end{pgfonlayer}
\end{tikzpicture}
=
\begin{tikzpicture}
	\begin{pgfonlayer}{nodelayer}
		\node [style=zmul] (0) at (-2, 0.5) {};
		\node [style=none] (1) at (-1, -0) {};
		\node [style=zmul] (2) at (-1, 1) {};
		\node [style=none] (3) at (-3, 0.5) {};
	\end{pgfonlayer}
	\begin{pgfonlayer}{edgelayer}
		\draw [style=simple, in=-27, out=180, looseness=1.00] (1.center) to (0);
		\draw [style=simple, in=180, out=27, looseness=1.00] (0) to (2.center);
		\draw [style=simple] (3.center) to (0);
	\end{pgfonlayer}
\end{tikzpicture}
=
\begin{tikzpicture}[yscale=-1]
	\begin{pgfonlayer}{nodelayer}
		\node [style=zmul] (0) at (-2, 0.5) {};
		\node [style=none] (1) at (-1, -0) {};
		\node [style=zmul] (2) at (-1, 1) {};
		\node [style=none] (3) at (-3, 0.5) {};
	\end{pgfonlayer}
	\begin{pgfonlayer}{edgelayer}
		\draw [style=simple, in=-27, out=180, looseness=1.00] (1.center) to (0);
		\draw [style=simple, in=180, out=27, looseness=1.00] (0) to (2.center);
		\draw [style=simple] (3.center) to (0);
	\end{pgfonlayer}
\end{tikzpicture}
$$

$$
\begin{tikzpicture}
	\begin{pgfonlayer}{nodelayer}
		\node [style=zmul] (0) at (0, -0) {};
		\node [style=none] (1) at (1, -0) {};
		\node [style=none] (2) at (-2, 1) {};
		\node [style=none] (3) at (-2, -0) {};
		\node [style=none] (4) at (-2, -0.75) {};
		\node [style=zmul] (5) at (-1, 0.5) {};
	\end{pgfonlayer}
	\begin{pgfonlayer}{edgelayer}
		\draw [style=simple, in=-159, out=0, looseness=1.00] (4.center) to (0);
		\draw [style=simple] (0) to (1.center);
		\draw [style=simple, in=0, out=153, looseness=1.00] (0) to (5);
		\draw [style=simple, in=0, out=153, looseness=1.00] (5) to (2.center);
		\draw [style=simple, in=-153, out=0, looseness=1.00] (3.center) to (5);
	\end{pgfonlayer}
\end{tikzpicture}
=
\begin{tikzpicture}[yscale=-1]
	\begin{pgfonlayer}{nodelayer}
		\node [style=zmul] (0) at (0, -0) {};
		\node [style=none] (1) at (1, -0) {};
		\node [style=none] (2) at (-2, 1) {};
		\node [style=none] (3) at (-2, -0) {};
		\node [style=none] (4) at (-2, -0.75) {};
		\node [style=zmul] (5) at (-1, 0.5) {};
	\end{pgfonlayer}
	\begin{pgfonlayer}{edgelayer}
		\draw [style=simple, in=-159, out=0, looseness=1.00] (4.center) to (0);
		\draw [style=simple] (0) to (1.center);
		\draw [style=simple, in=0, out=153, looseness=1.00] (0) to (5);
		\draw [style=simple, in=0, out=153, looseness=1.00] (5) to (2.center);
		\draw [style=simple, in=-153, out=0, looseness=1.00] (3.center) to (5);
	\end{pgfonlayer}
\end{tikzpicture}
\hspace*{.5cm}
\begin{tikzpicture}[xscale=-1]
	\begin{pgfonlayer}{nodelayer}
		\node [style=zmul] (0) at (0, -0) {};
		\node [style=none] (1) at (1, -0) {};
		\node [style=none] (2) at (-2, 1) {};
		\node [style=none] (3) at (-2, -0) {};
		\node [style=none] (4) at (-2, -0.75) {};
		\node [style=zmul] (5) at (-1, 0.5) {};
	\end{pgfonlayer}
	\begin{pgfonlayer}{edgelayer}
		\draw [style=simple, in=-159, out=0, looseness=1.00] (4.center) to (0);
		\draw [style=simple] (0) to (1.center);
		\draw [style=simple, in=0, out=153, looseness=1.00] (0) to (5);
		\draw [style=simple, in=0, out=153, looseness=1.00] (5) to (2.center);
		\draw [style=simple, in=-153, out=0, looseness=1.00] (3.center) to (5);
	\end{pgfonlayer}
\end{tikzpicture}
=
\begin{tikzpicture}[yscale=-1, xscale=-1]
	\begin{pgfonlayer}{nodelayer}
		\node [style=zmul] (0) at (0, -0) {};
		\node [style=none] (1) at (1, -0) {};
		\node [style=none] (2) at (-2, 1) {};
		\node [style=none] (3) at (-2, -0) {};
		\node [style=none] (4) at (-2, -0.75) {};
		\node [style=zmul] (5) at (-1, 0.5) {};
	\end{pgfonlayer}
	\begin{pgfonlayer}{edgelayer}
		\draw [style=simple, in=-159, out=0, looseness=1.00] (4.center) to (0);
		\draw [style=simple] (0) to (1.center);
		\draw [style=simple, in=0, out=153, looseness=1.00] (0) to (5);
		\draw [style=simple, in=0, out=153, looseness=1.00] (5) to (2.center);
		\draw [style=simple, in=-153, out=0, looseness=1.00] (3.center) to (5);
	\end{pgfonlayer}
\end{tikzpicture}
$$

And the Frobenius law holds:
 \begin{description}
 \item[{[F]}]
 \hfil
$
\begin{tikzpicture}
	\begin{pgfonlayer}{nodelayer}
		\node [style=zmul] (0) at (4, -0) {};
		\node [style=zmul] (1) at (3, 1) {};
		\node [style=rn] (2) at (5, 1) {};
		\node [style=rn] (3) at (5, -0) {};
		\node [style=rn] (4) at (2, -0) {};
		\node [style=rn] (5) at (2, 1) {};
	\end{pgfonlayer}
	\begin{pgfonlayer}{edgelayer}
		\draw [style=simple, in=150, out=-30, looseness=1.00] (1) to (0);
		\draw [style=simple, in=180, out=30, looseness=1.00] (1) to (2);
		\draw [style=simple] (3) to (0);
		\draw [style=simple, in=0, out=-150, looseness=1.00] (0) to (4);
		\draw [style=simple] (5) to (1);
	\end{pgfonlayer}
\end{tikzpicture}
=
\begin{tikzpicture}
	\begin{pgfonlayer}{nodelayer}
		\node [style=rn] (0) at (5, 1) {};
		\node [style=rn] (1) at (5, -0) {};
		\node [style=rn] (2) at (2, -0) {};
		\node [style=rn] (3) at (2, 1) {};
		\node [style=zmul] (4) at (3, 0.5) {};
		\node [style=zmul] (5) at (4, 0.5) {};
	\end{pgfonlayer}
	\begin{pgfonlayer}{edgelayer}
		\draw [style=simple, in=-27, out=180, looseness=1.00] (1) to (5);
		\draw [style=simple, in=180, out=27, looseness=1.00] (5) to (0);
		\draw [style=simple] (5) to (4);
		\draw [style=simple, in=0, out=153, looseness=1.00] (4) to (3);
		\draw [style=simple, in=0, out=-153, looseness=1.00] (4) to (2);
	\end{pgfonlayer}
\end{tikzpicture}
=
\begin{tikzpicture}[xscale=-1]
	\begin{pgfonlayer}{nodelayer}
		\node [style=zmul] (0) at (4, -0) {};
		\node [style=zmul] (1) at (3, 1) {};
		\node [style=rn] (2) at (5, 1) {};
		\node [style=rn] (3) at (5, -0) {};
		\node [style=rn] (4) at (2, -0) {};
		\node [style=rn] (5) at (2, 1) {};
	\end{pgfonlayer}
	\begin{pgfonlayer}{edgelayer}
		\draw [style=simple, in=150, out=-30, looseness=1.00] (1) to (0);
		\draw [style=simple, in=180, out=30, looseness=1.00] (1) to (2);
		\draw [style=simple] (3) to (0);
		\draw [style=simple, in=0, out=-150, looseness=1.00] (0) to (4);
		\draw [style=simple] (5) to (1);
	\end{pgfonlayer}
\end{tikzpicture}
$
 \end{description}

 A Frobenius algebra is {\bf special} if 
 $$
 \begin{tikzpicture}
	\begin{pgfonlayer}{nodelayer}
		\node [style=zmul] (0) at (2, 3) {};
		\node [style=zmul] (1) at (3, 3) {};
		\node [style=rn] (2) at (4, 3) {};
		\node [style=rn] (3) at (1, 3) {};
	\end{pgfonlayer}
	\begin{pgfonlayer}{edgelayer}
		\draw [style=simple] (3) to (0);
		\draw [style=simple, bend left=45, looseness=1.25] (0) to (1);
		\draw [style=simple, bend left=45, looseness=1.25] (1) to (0);
		\draw [style=simple] (1) to (2);
	\end{pgfonlayer}
\end{tikzpicture}
=
\begin{tikzpicture}
	\begin{pgfonlayer}{nodelayer}
		\node [style=rn] (0) at (2, 3) {};
		\node [style=rn] (1) at (1, 3) {};
	\end{pgfonlayer}
	\begin{pgfonlayer}{edgelayer}
		\draw [style=simple] (0) to (1);
	\end{pgfonlayer}
\end{tikzpicture}
 $$
and {\bf commutative} if the underlying monoid and comonoids are commutative and cocommutative:

$$
\begin{tikzpicture}
	\begin{pgfonlayer}{nodelayer}
		\node [style=zmul] (0) at (-2, 0.5) {};
		\node [style=none] (1) at (-1, -0) {};
		\node [style=none] (2) at (-1, 1) {};
		\node [style=none] (3) at (-3, 0.5) {};
		\node [style=none] (4) at (0, 1) {};
		\node [style=none] (5) at (0, -0) {};
	\end{pgfonlayer}
	\begin{pgfonlayer}{edgelayer}
		\draw [style=simple, in=-27, out=180, looseness=1.00] (1.center) to (0);
		\draw [style=simple, in=180, out=27, looseness=1.00] (0) to (2.center);
		\draw [style=simple] (3.center) to (0);
		\draw [style=simple, in=180, out=0, looseness=0.75] (1.center) to (4.center);
		\draw [style=simple, in=0, out=180, looseness=0.75] (5.center) to (2.center);
	\end{pgfonlayer}
\end{tikzpicture}
=
\begin{tikzpicture}
	\begin{pgfonlayer}{nodelayer}
		\node [style=zmul] (0) at (-2, 0.5) {};
		\node [style=none] (1) at (-1, -0) {};
		\node [style=none] (2) at (-1, 1) {};
		\node [style=none] (3) at (-3, 0.5) {};
	\end{pgfonlayer}
	\begin{pgfonlayer}{edgelayer}
		\draw [style=simple, in=-27, out=180, looseness=1.00] (1.center) to (0);
		\draw [style=simple, in=180, out=27, looseness=1.00] (0) to (2.center);
		\draw [style=simple] (3.center) to (0);
	\end{pgfonlayer}
\end{tikzpicture}
\hspace*{.5cm}
\begin{tikzpicture}[xscale=-1]
	\begin{pgfonlayer}{nodelayer}
		\node [style=zmul] (0) at (-2, 0.5) {};
		\node [style=none] (1) at (-1, -0) {};
		\node [style=none] (2) at (-1, 1) {};
		\node [style=none] (3) at (-3, 0.5) {};
		\node [style=none] (4) at (0, 1) {};
		\node [style=none] (5) at (0, -0) {};
	\end{pgfonlayer}
	\begin{pgfonlayer}{edgelayer}
		\draw [style=simple, in=-27, out=180, looseness=1.00] (1.center) to (0);
		\draw [style=simple, in=180, out=27, looseness=1.00] (0) to (2.center);
		\draw [style=simple] (3.center) to (0);
		\draw [style=simple, in=180, out=0, looseness=0.75] (1.center) to (4.center);
		\draw [style=simple, in=0, out=180, looseness=0.75] (5.center) to (2.center);
	\end{pgfonlayer}
\end{tikzpicture}
=
\begin{tikzpicture}[xscale=-1]
	\begin{pgfonlayer}{nodelayer}
		\node [style=zmul] (0) at (-2, 0.5) {};
		\node [style=none] (1) at (-1, -0) {};
		\node [style=none] (2) at (-1, 1) {};
		\node [style=none] (3) at (-3, 0.5) {};
	\end{pgfonlayer}
	\begin{pgfonlayer}{edgelayer}
		\draw [style=simple, in=-27, out=180, looseness=1.00] (1.center) to (0);
		\draw [style=simple, in=180, out=27, looseness=1.00] (0) to (2.center);
		\draw [style=simple] (3.center) to (0);
	\end{pgfonlayer}
\end{tikzpicture}
$$

 A {\bf \dag-Frobenius algebra} $(A,\zmul,\zmulu)$ is a Frobenius algebra of the form 
 $$(A,\zmul,\zmulu,\zmul^\dag,\zmulu^\dag)$$
 That is to say, the monoid and comonoid are daggers of each other.
 
Special commutative \dag-Frobenius algebras are called {\bf classical structures}.

 A non-(co)unital  special commutative \dag-Frobenius algebra  is called  a {\bf semi-Frobenius algebra}. Semi-Frobenius algebras are used to construct a weak product structure for inverse categories such as $\CNOT$ and $\TOF$.

\end{definition}

 However, we are interested in a simple fragment of the ZX-calculus, namely the angle-free calculus for real stabilizer circuits, $\ZX_\pi$, described in \cite{realzx}  (slightly modified to account for scalars):
 \begin{definition}
\label{def:ZX.pi} 

Let $\ZX_\pi $ denote the $\dag$-compact closed PROP with generators:

$$
\zmul \hspace*{.5cm} \zmulu \hspace*{.5cm} \zcomul \hspace*{.5cm} \zcomulu \hspace*{.5cm} \hadamard
$$
such that
$$( \zmul, \zmulu, \zcomul, \zcomulu)$$
is a classical structure, corresponding to the $Z$ basis, 
and the following identities also hold up to swapping colours:

\begin{figure}[H]
	\noindent
	\scalebox{1.0}{%
		\vbox{%
			\begin{mdframed}
				\begin{multicols}{2}
					\begin{description}
						\item[{[PP]}]
						\namedlabel{ZX.pi:PP}{\bf [PP]}
						\hfil
						$
	\begin{tikzpicture}
	\begin{pgfonlayer}{nodelayer}
		\node [style=xmul] (0) at (8, -0) {$\pi$};
		\node [style=xmul] (1) at (9.5, -0) {$\pi$};
		\node [style=none] (2) at (10.5, -0) {};
		\node [style=none] (3) at (7, -0) {};
	\end{pgfonlayer}	
	\begin{pgfonlayer}{edgelayer}
		\draw (0) to (1);
		\draw (2.center) to (1);
		\draw (0) to (3.center);
	\end{pgfonlayer}
\end{tikzpicture}
						=
						\begin{tikzpicture}
	\begin{pgfonlayer}{nodelayer}
		\node [style=rn] (0) at (8, -0) {};
		\node [style=rn] (1) at (9.5, -0) {};
	\end{pgfonlayer}
	\begin{pgfonlayer}{edgelayer}
		\draw (0) to (1);
	\end{pgfonlayer}
\end{tikzpicture}
						$

						\item[{[PI]}]
						\namedlabel{ZX.pi:PI}{\bf [PI]}
						\hfil
						$
\begin{tikzpicture}
	\begin{pgfonlayer}{nodelayer}
		\node [style=rn] (0) at (4.25, 2.25) {};
		\node [style=rn] (1) at (0.75, 3) {};
		\node [style=zmul] (2) at (3, 3) {};
		\node [style=xmul] (3) at (1.5, 3) {$\pi$};
		\node [style=rn] (4) at (4.25, 3.75) {};
	\end{pgfonlayer}
	\begin{pgfonlayer}{edgelayer}
		\draw [style=simple] (1) to (3);
		\draw [style=simple] (3) to (2);
		\draw [style=simple, in=180, out=-37, looseness=1.00] (2) to (0);
		\draw [style=simple, in=37, out=180, looseness=1.00] (4) to (2);
	\end{pgfonlayer}
\end{tikzpicture}
=
\begin{tikzpicture}
	\begin{pgfonlayer}{nodelayer}
		\node [style=rn] (0) at (4.5, 2.25) {};
		\node [style=rn] (1) at (1, 3) {};
		\node [style=rn] (2) at (4.5, 3.75) {};
		\node [style=zmul] (3) at (2, 3) {};
		\node [style=xmul] (4) at (3.5, 3.75) {$\pi$};
		\node [style=xmul] (5) at (3.5, 2.25) {$\pi$};
	\end{pgfonlayer}
	\begin{pgfonlayer}{edgelayer}
		\draw [style=simple] (1) to (3);
		\draw [style=simple] (5) to (0);
		\draw [style=simple] (4) to (2);
		\draw [style=simple, in=37, out=180, looseness=1.00] (4) to (3);
		\draw [style=simple, in=-37, out=180, looseness=1.00] (5) to (3);
	\end{pgfonlayer}
\end{tikzpicture}$

				\hfil$
\begin{tikzpicture}
	\begin{pgfonlayer}{nodelayer}
		\node [style=rn] (0) at (0.75, 2.25) {};
		\node [style=rn] (1) at (4.25, 3) {};
		\node [style=zmul] (2) at (2, 3) {};
		\node [style=xmul] (3) at (3.5, 3) {$\pi$};
		\node [style=rn] (4) at (0.75, 3.75) {};
	\end{pgfonlayer}
	\begin{pgfonlayer}{edgelayer}
		\draw [style=simple] (1) to (3);
		\draw [style=simple] (3) to (2);
		\draw [style=simple, in=0, out=-143, looseness=1.00] (2) to (0);
		\draw [style=simple, in=143, out=0, looseness=1.00] (4) to (2);
	\end{pgfonlayer}
\end{tikzpicture}
=
\begin{tikzpicture}
	\begin{pgfonlayer}{nodelayer}
		\node [style=rn] (0) at (1, 2.25) {};
		\node [style=rn] (1) at (4.5, 3) {};
		\node [style=rn] (2) at (1, 3.75) {};
		\node [style=zmul] (3) at (3.5, 3) {};
		\node [style=xmul] (4) at (2, 3.75) {$\pi$};
		\node [style=xmul] (5) at (2, 2.25) {$\pi$};
	\end{pgfonlayer}
	\begin{pgfonlayer}{edgelayer}
		\draw [style=simple] (1) to (3);
		\draw [style=simple] (5) to (0);
		\draw [style=simple] (4) to (2);
		\draw [style=simple, in=143, out=0, looseness=1.00] (4) to (3);
		\draw [style=simple, in=-143, out=0, looseness=1.00] (5) to (3);
	\end{pgfonlayer}
\end{tikzpicture}$
						
						\item[{[B.U']}]
						\namedlabel{ZX.pi:B.U}{\bf [B.U']}
						\hfil
						$
\begin{tikzpicture}
	\begin{pgfonlayer}{nodelayer}
		\node [style=rn] (0) at (5, 4.5) {};
		\node [style=xmul] (1) at (3, 4) {};
		\node [style=zmul] (2) at (4, 4) {};
		\node [style=rn] (3) at (5, 3.5) {};
	\end{pgfonlayer}
	\begin{pgfonlayer}{edgelayer}
		\draw [style=simple, in=-27, out=180, looseness=1.00] (3) to (2);
		\draw [style=simple] (2) to (1);
		\draw [style=simple, in=180, out=27, looseness=1.00] (2) to (0);
	\end{pgfonlayer}
\end{tikzpicture}
=
\begin{tikzpicture}
	\begin{pgfonlayer}{nodelayer}
		\node [style=rn] (0) at (5.25, 4.5) {};
		\node [style=xmul] (1) at (4, 4.5) {};
		\node [style=rn] (2) at (5.25, 3.75) {};
		\node [style=xmul] (3) at (4, 3.75) {};
		\node [style=xmul] (4) at (4, 5.5) {};
		\node [style=zmul] (5) at (5, 5.5) {};
	\end{pgfonlayer}
	\begin{pgfonlayer}{edgelayer}
		\draw [style=simple] (3) to (2);
		\draw [style=simple] (0) to (1);
		\draw [style=simple, in=-45, out=-135, looseness=1.25] (5) to (4);
		\draw [style=simple, in=135, out=45, looseness=1.25] (4) to (5);
		\draw [style=simple] (5) to (4);
	\end{pgfonlayer}
\end{tikzpicture}
$ 

\hfil
$
\begin{tikzpicture}[xscale=-1]
	\begin{pgfonlayer}{nodelayer}
		\node [style=rn] (0) at (5, 4.5) {};
		\node [style=xmul] (1) at (3, 4) {};
		\node [style=zmul] (2) at (4, 4) {};
		\node [style=rn] (3) at (5, 3.5) {};
	\end{pgfonlayer}
	\begin{pgfonlayer}{edgelayer}
		\draw [style=simple, in=-27, out=180, looseness=1.00] (3) to (2);
		\draw [style=simple] (2) to (1);
		\draw [style=simple, in=180, out=27, looseness=1.00] (2) to (0);
	\end{pgfonlayer}
\end{tikzpicture}
=
\begin{tikzpicture}[xscale=-1]
	\begin{pgfonlayer}{nodelayer}
		\node [style=rn] (0) at (5.25, 4.5) {};
		\node [style=xmul] (1) at (4, 4.5) {};
		\node [style=rn] (2) at (5.25, 3.75) {};
		\node [style=xmul] (3) at (4, 3.75) {};
		\node [style=xmul] (4) at (4, 5.5) {};
		\node [style=zmul] (5) at (5, 5.5) {};
	\end{pgfonlayer}
	\begin{pgfonlayer}{edgelayer}
		\draw [style=simple] (3) to (2);
		\draw [style=simple] (0) to (1);
		\draw [style=simple, in=-45, out=-135, looseness=1.25] (5) to (4);
		\draw [style=simple, in=135, out=45, looseness=1.25] (4) to (5);
		\draw [style=simple] (5) to (4);
	\end{pgfonlayer}
\end{tikzpicture}
$

	\item[{[B.H']}]
	\namedlabel{ZX.pi:B.H}{\bf [B.H']}
	\hfil$
\begin{tikzpicture}
	\begin{pgfonlayer}{nodelayer}
		\node [style=rn] (0) at (5, 5) {};
		\node [style=xmul] (1) at (4, 5) {};
		\node [style=zmul] (2) at (3, 5) {};
		\node [style=rn] (3) at (2, 5) {};
	\end{pgfonlayer}
	\begin{pgfonlayer}{edgelayer}
		\draw [style=simple] (0) to (1);
		\draw [style=simple, bend right=45, looseness=1.25] (1) to (2);
		\draw [style=simple] (3) to (2);
		\draw [style=simple, bend right=45, looseness=1.25] (2) to (1);
	\end{pgfonlayer}
\end{tikzpicture}
=
\begin{tikzpicture}
	\begin{pgfonlayer}{nodelayer}
		\node [style=rn] (0) at (2, 4.5) {};
		\node [style=zmul] (1) at (3, 4.5) {};
		\node [style=rn] (2) at (5, 4.5) {};
		\node [style=xmul] (3) at (4, 4.5) {};
		\node [style=xmul] (4) at (2, 5.5) {};
		\node [style=zmul] (5) at (3, 5.5) {};
		\node [style=zmul] (6) at (5, 5.5) {};
		\node [style=xmul] (7) at (4, 5.5) {};
	\end{pgfonlayer}
	\begin{pgfonlayer}{edgelayer}
		\draw [style=simple] (0) to (1);
		\draw [style=simple] (2) to (3);
		\draw [style=simple, in=-45, out=-135, looseness=1.25] (5) to (4);
		\draw [style=simple, in=135, out=45, looseness=1.25] (4) to (5);
		\draw [style=simple] (5) to (4);
		\draw [style=simple, in=-45, out=-135, looseness=1.25] (6) to (7);
		\draw [style=simple, in=135, out=45, looseness=1.25] (7) to (6);
		\draw [style=simple] (6) to (7);
	\end{pgfonlayer}
\end{tikzpicture}$

						\item[{[B.M']}]
						\namedlabel{ZX.pi:B.M}{\bf [B.M']}
						\hfil$
\begin{tikzpicture}
	\begin{pgfonlayer}{nodelayer}
		\node [style=rn] (0) at (-2, 3) {};
		\node [style=xmul] (1) at (-1, 3) {};
		\node [style=xmul] (2) at (-1, 2) {};
		\node [style=zmul] (3) at (1, 3) {};
		\node [style=zmul] (4) at (1, 2) {};
		\node [style=rn] (5) at (2, 2) {};
		\node [style=rn] (6) at (-2, 2) {};
		\node [style=rn] (7) at (2, 3) {};
	\end{pgfonlayer}
	\begin{pgfonlayer}{edgelayer}
		\draw [style=simple] (6) to (2);
		\draw [style=simple, bend right, looseness=1.00] (2) to (4);
		\draw [style=simple] (4) to (5);
		\draw [style=simple] (3) to (2);
		\draw [style=simple] (4) to (1);
		\draw [style=simple, bend left, looseness=1.00] (1) to (3);
		\draw [style=simple] (0) to (1);
		\draw [style=simple] (7) to (3);
	\end{pgfonlayer}
\end{tikzpicture}
=
\begin{tikzpicture}
	\begin{pgfonlayer}{nodelayer}
		\node [style=rn] (0) at (-1, 3) {};
		\node [style=rn] (1) at (2, 2) {};
		\node [style=rn] (2) at (-1, 2) {};
		\node [style=xmul] (3) at (1, 2.5) {};
		\node [style=zmul] (4) at (0, 2.5) {};
		\node [style=rn] (5) at (2, 3) {};
		\node [style=xmul] (6) at (1, 3.75) {};
		\node [style=zmul] (7) at (0, 3.75) {};
	\end{pgfonlayer}
	\begin{pgfonlayer}{edgelayer}
		\draw [style=simple, in=0, out=165, looseness=1.00] (4) to (0);
		\draw [style=simple, in=-165, out=0, looseness=1.00] (2) to (4);
		\draw [style=simple] (3) to (4);
		\draw [style=simple, in=-15, out=180, looseness=1.00] (1) to (3);
		\draw [style=simple, in=180, out=15, looseness=1.00] (3) to (5);
		\draw [style=simple, in=-45, out=-135, looseness=1.25] (6) to (7);
		\draw [style=simple, in=45, out=135, looseness=1.25] (6) to (7);
		\draw [style=simple] (7) to (6);
	\end{pgfonlayer}
\end{tikzpicture}$
						
						\item[{[H2]}]
						\namedlabel{ZX.pi:H2}{\bf [H2]}
						\hfil$
						\begin{tikzpicture}
	\begin{pgfonlayer}{nodelayer}
		\node [style=rn] (0) at (-2, 3) {};
		\node [style=rn] (1) at (1, 3) {};
		\node [style=h] (2) at (-1, 3) {};
		\node [style=h] (3) at (0, 3) {};
	\end{pgfonlayer}
	\begin{pgfonlayer}{edgelayer}
		\draw [style=simple] (1) to (3);
		\draw [style=simple] (3) to (2);
		\draw [style=simple] (2) to (0);
	\end{pgfonlayer}
\end{tikzpicture}=
\begin{tikzpicture}
	\begin{pgfonlayer}{nodelayer}
		\node [style=rn] (0) at (-2, 3) {};
		\node [style=rn] (1) at (1, 3) {};
	\end{pgfonlayer}
	\begin{pgfonlayer}{edgelayer}
		\draw (1) to (0);
	\end{pgfonlayer}
\end{tikzpicture}
$

\item[{[ZO]}]
\namedlabel{ZX.pi:ZO}{\bf [ZO]}
\hfil$
\begin{tikzpicture}
	\begin{pgfonlayer}{nodelayer}
		\node [style=xmul] (2) at (4, 1) {$\pi$};
		\node [style=rn] (3) at (5, -0) {};
		\node [style=rn] (4) at (3, -0) {};
	\end{pgfonlayer}
	\begin{pgfonlayer}{edgelayer}
		\draw [style=simple] (3) to (4);
	\end{pgfonlayer}
\end{tikzpicture}
=
\begin{tikzpicture}
	\begin{pgfonlayer}{nodelayer}
		\node [style=xmul] (2) at (4.5, 1) {$\pi$};
		\node [style=rn] (3) at (3, -0) {};
		\node [style=rn] (4) at (6, -0) {};
		\node [style=zmul] (5) at (4, -0) {};
		\node [style=xmul] (6) at (5, -0) {};
	\end{pgfonlayer}
	\begin{pgfonlayer}{edgelayer}
		\draw [style=simple] (4) to (6);
		\draw [style=simple] (5) to (3);
	\end{pgfonlayer}
\end{tikzpicture}
$

\item[{[IV]}]
\namedlabel{ZX.pi:IV}{\bf [IV]}
\hfil
$
		\begin{tikzpicture}
	\begin{pgfonlayer}{nodelayer}
		\node [style=zmul] (0) at (0, 2.75) {};
		\node [style=xmul] (1) at (2, 2.75) {};
		\node [style=xmul] (2) at (2, 1.75) {};
		\node [style=zmul] (3) at (0, 1.75) {};
	\end{pgfonlayer}
	\begin{pgfonlayer}{edgelayer}
		\draw (3) to (2);
		\draw [bend right, looseness=1.25] (1) to (0);
		\draw [bend right, looseness=1.25] (0) to (1);
		\draw (1) to (0);
	\end{pgfonlayer}
\end{tikzpicture}
		=
		\begin{tikzpicture}
	\begin{pgfonlayer}{nodelayer}
		\node [style=rn] (0) at (0, 2.75) {};
		\node [style=rn] (1) at (2, 2.75) {};
		\node [style=rn] (2) at (2, 1.75) {};
		\node [style=rn] (3) at (0, 1.75) {};
	\end{pgfonlayer}
	\begin{pgfonlayer}{edgelayer}
	\end{pgfonlayer}
\end{tikzpicture}
		$

\item[{[L]}]
\namedlabel{ZX.pi:L}{\bf [L]}
\hfil
$
\begin{tikzpicture}
	\begin{pgfonlayer}{nodelayer}
		\node [style=rn] (0) at (-2, 1) {};
		\node [style=rn] (1) at (-2, 2) {};
	\end{pgfonlayer}
	\begin{pgfonlayer}{edgelayer}
		\draw [bend left=90, looseness=2.50] (1) to (0);
	\end{pgfonlayer}
\end{tikzpicture}
=
\begin{tikzpicture}
	\begin{pgfonlayer}{nodelayer}
		\node [style=rn] (0) at (-2, 1) {};
		\node [style=rn] (1) at (-2, 2) {};
		\node [style=xmul] (2) at (-1, 1.5) {};
	\end{pgfonlayer}
	\begin{pgfonlayer}{edgelayer}
		\draw [in=90, out=0, looseness=1] (1) to (2);
		\draw [in=0, out=-90, looseness=1] (2) to (0);
	\end{pgfonlayer}
\end{tikzpicture}
=
\begin{tikzpicture}
	\begin{pgfonlayer}{nodelayer}
		\node [style=rn] (0) at (-2, 1) {};
		\node [style=rn] (1) at (-2, 2) {};
		\node [style=zmul] (2) at (-1, 1.5) {};
	\end{pgfonlayer}
	\begin{pgfonlayer}{edgelayer}
		\draw [in=90, out=0, looseness=1] (1) to (2);
		\draw [in=0, out=-90, looseness=1] (2) to (0);
	\end{pgfonlayer}
\end{tikzpicture}
$

\hfil$
\begin{tikzpicture}[xscale=-1]
	\begin{pgfonlayer}{nodelayer}
		\node [style=rn] (0) at (-2, 1) {};
		\node [style=rn] (1) at (-2, 2) {};
	\end{pgfonlayer}
	\begin{pgfonlayer}{edgelayer}
		\draw [bend left=90, looseness=2.50] (1) to (0);
	\end{pgfonlayer}
\end{tikzpicture}
=
\begin{tikzpicture}[xscale=-1]
	\begin{pgfonlayer}{nodelayer}
		\node [style=rn] (0) at (-2, 1) {};
		\node [style=rn] (1) at (-2, 2) {};
		\node [style=xmul] (2) at (-1, 1.5) {};
	\end{pgfonlayer}
	\begin{pgfonlayer}{edgelayer}
		\draw [in=90, out=0, looseness=1] (1) to (2);
		\draw [in=0, out=-90, looseness=1] (2) to (0);
	\end{pgfonlayer}
\end{tikzpicture}
=
\begin{tikzpicture}[xscale=-1]
	\begin{pgfonlayer}{nodelayer}
		\node [style=rn] (0) at (-2, 1) {};
		\node [style=rn] (1) at (-2, 2) {};
		\node [style=zmul] (2) at (-1, 1.5) {};
	\end{pgfonlayer}
	\begin{pgfonlayer}{edgelayer}
		\draw [in=90, out=0, looseness=1] (1) to (2);
		\draw [in=0, out=-90, looseness=1] (2) to (0);
	\end{pgfonlayer}
\end{tikzpicture}
$

\item[{[S1]}]
\namedlabel{ZX.pi:S1}{\bf [S1]}
\hfil
$
\begin{tikzpicture}
	\begin{pgfonlayer}{nodelayer}
		\node [style=zmul] (0) at (4, -0) {$\pi$};
		\node [style=rn] (1) at (3, -0) {};
		\node [style=rn] (2) at (5, -0) {};
	\end{pgfonlayer}
	\begin{pgfonlayer}{edgelayer}
		\draw [style=simple] (2) to (0);
		\draw [style=simple] (0) to (1);
	\end{pgfonlayer}
\end{tikzpicture}
:=
\begin{tikzpicture}
	\begin{pgfonlayer}{nodelayer}
		\node [style=h] (0) at (4.5, 1) {};
		\node [style=rn] (1) at (3, -0) {};
		\node [style=rn] (2) at (6, -0) {};
		\node [style=zmul] (3) at (4.5, -0) {};
		\node [style=zmul] (4) at (3.75, -0.75) {};
		\node [style=xmul] (5) at (5.25, -0.75) {};
	\end{pgfonlayer}
	\begin{pgfonlayer}{edgelayer}
		\draw [style=simple, in=0, out=30, looseness=1.25] (3) to (0);
		\draw [style=simple, in=150, out=180, looseness=1.25] (0) to (3);
		\draw [style=simple] (2) to (3);
		\draw [style=simple] (3) to (1);
		\draw (5) to (4);
	\end{pgfonlayer}
\end{tikzpicture}
$

\item[{[S2]}]
\namedlabel{ZX.pi:S2}{\bf [S2]}
\hfil
$
\begin{tikzpicture}
	\begin{pgfonlayer}{nodelayer}
		\node [style=rn] (0) at (1, 2) {$\vdots$};
		\node [style=xmul] (1) at (0, 2) {$\alpha$};
		\node [style=rn] (2) at (-1, 2) {$\vdots$};
		\node [style=rn] (3) at (1, 3) {};
		\node [style=rn] (4) at (1, 1) {};
		\node [style=rn] (5) at (-1, 1) {};
		\node [style=rn] (6) at (-1, 3) {};
	\end{pgfonlayer}
	\begin{pgfonlayer}{edgelayer}
		\draw [style=simple, in=135, out=0, looseness=1.00] (6) to (1);
		\draw [style=simple, in=180, out=45, looseness=1.00] (1) to (3);
		\draw [style=simple, in=-45, out=180, looseness=1.00] (4) to (1);
		\draw [style=simple, in=0, out=-135, looseness=1.00] (1) to (5);
	\end{pgfonlayer}
\end{tikzpicture}
						:=
							\begin{tikzpicture}
	\begin{pgfonlayer}{nodelayer}
		\node [style=rn] (0) at (1.5, 2) {$\vdots$};
		\node [style=zmul] (1) at (0, 2) {$\alpha$};
		\node [style=rn] (2) at (2, 1) {};
		\node [style=rn] (3) at (2, 3) {};
		\node [style=rn] (4) at (-2, 1) {};
		\node [style=rn] (5) at (-2, 3) {};
		\node [style=rn] (6) at (-1.5, 2) {$\vdots$};
		\node [style=h] (7) at (1, 3) {};
		\node [style=h] (8) at (1, 1) {};
		\node [style=h] (9) at (-1, 1) {};
		\node [style=h] (10) at (-1, 3) {};
	\end{pgfonlayer}
	\begin{pgfonlayer}{edgelayer}
		\draw [style=simple] (3) to (7);
		\draw [style=simple] (10) to (5);
		\draw [style=simple] (4) to (9);
		\draw [style=simple] (2) to (8);
		\draw [style=simple, in=135, out=0, looseness=1.00] (10) to (1);
		\draw [style=simple, in=180, out=45, looseness=1.00] (1) to (7);
		\draw [style=simple, in=-45, out=180, looseness=1.00] (8) to (1);
		\draw [style=simple, in=0, out=-135, looseness=1.00] (1) to (9);
	\end{pgfonlayer}
\end{tikzpicture}
$

\item[{[S3]}]
\namedlabel{ZX.pi:S3}{\bf [S3]}
\hfil
$
\begin{tikzpicture}
	\begin{pgfonlayer}{nodelayer}
		\node [style=rn] (0) at (1, 2) {$\vdots$};
		\node [style=zmul] (1) at (0, 2) {$0$};
		\node [style=rn] (2) at (-1, 2) {$\vdots$};
		\node [style=rn] (3) at (1, 3) {};
		\node [style=rn] (4) at (1, 1) {};
		\node [style=rn] (5) at (-1, 1) {};
		\node [style=rn] (6) at (-1, 3) {};
	\end{pgfonlayer}
	\begin{pgfonlayer}{edgelayer}
		\draw [style=simple, in=135, out=0, looseness=1.00] (6) to (1);
		\draw [style=simple, in=180, out=45, looseness=1.00] (1) to (3);
		\draw [style=simple, in=-45, out=180, looseness=1.00] (4) to (1);
		\draw [style=simple, in=0, out=-135, looseness=1.00] (1) to (5);
	\end{pgfonlayer}
\end{tikzpicture}
						:=
							\begin{tikzpicture}
	\begin{pgfonlayer}{nodelayer}
		\node [style=rn] (0) at (1.5, 2) {$\vdots$};
		\node [style=zmul] (1) at (0, 2) {$$};
		\node [style=rn] (2) at (2, 1) {};
		\node [style=rn] (3) at (2, 3) {};
		\node [style=rn] (4) at (-2, 1) {};
		\node [style=rn] (5) at (-2, 3) {};
		\node [style=rn] (6) at (-1.5, 2) {$\vdots$};
		\node [style=rn] (7) at (1, 3) {};
		\node [style=rn] (8) at (1, 1) {};
		\node [style=rn] (9) at (-1, 1) {};
		\node [style=rn] (10) at (-1, 3) {};
	\end{pgfonlayer}
	\begin{pgfonlayer}{edgelayer}
		\draw [style=simple] (3) to (7);
		\draw [style=simple] (10) to (5);
		\draw [style=simple] (4) to (9);
		\draw [style=simple] (2) to (8);
		\draw [style=simple, in=135, out=0, looseness=1.00] (10) to (1);
		\draw [style=simple, in=180, out=45, looseness=1.00] (1) to (7);
		\draw [style=simple, in=-45, out=180, looseness=1.00] (8) to (1);
		\draw [style=simple, in=0, out=-135, looseness=1.00] (1) to (9);
	\end{pgfonlayer}
\end{tikzpicture}
$

					\end{description}
				\end{multicols}
				\
			\end{mdframed}
	}}
	\caption{The identities of $\ZX_\pi$ (where $\alpha \in \{0,\pi\}$)}
	\label{fig:ZXPI}
\end{figure}
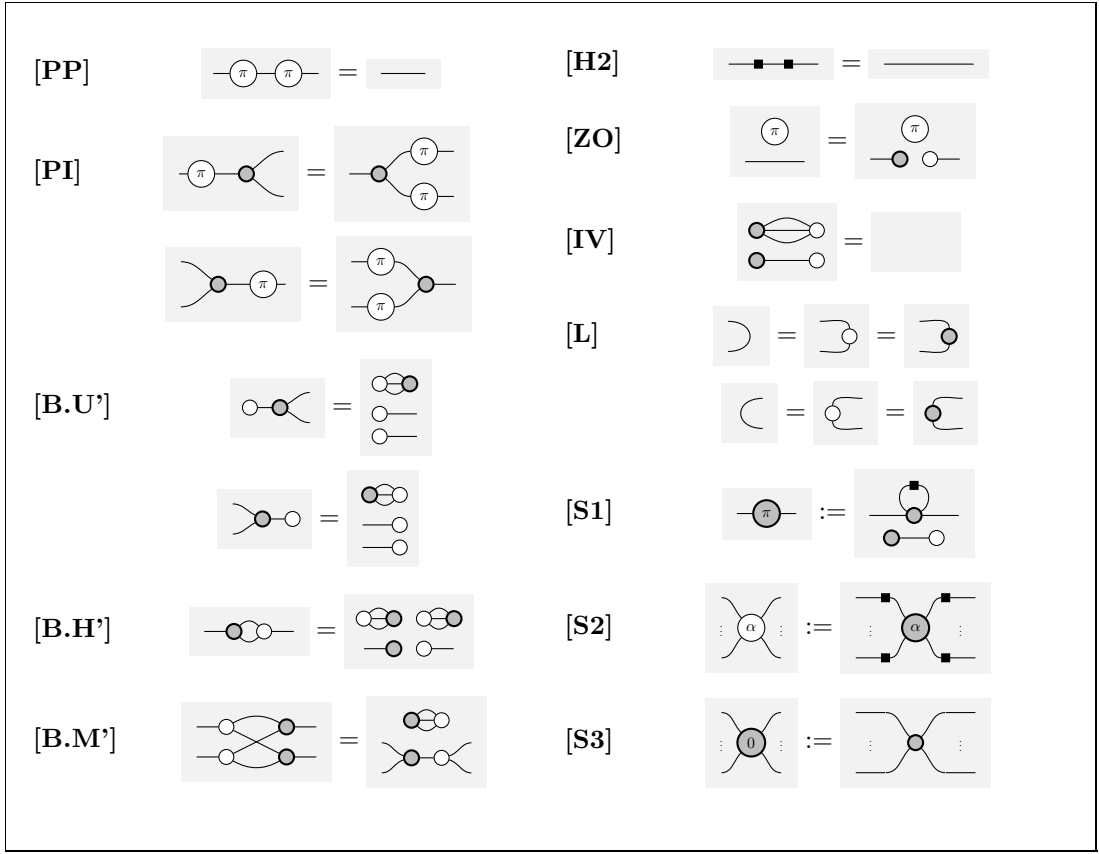

The last 3 Axioms are actually definitions, which simplify the presentation of $\ZX_\pi$.
Note that the axioms of a classical structure are omitted from this box to save space.

These axioms imply that the black and white Frobenius algebras are complementary where the antipode is the identity.

This category has a canonical \dag-functor, as all of the stated axioms are horizontally symmetric. It is also $\dag$-compact closed.

This category  embeds $\FdHilb$; the black Frobenius algebra corresponds to the Pauli $Z$ basis; the  white Frobenius algebra corresponds to the Pauli $X$ basis; the gate $\hadamard$ corresponds to the Hadamard gate and $\zpi$ and $\xpi$ correspond to $Z$ and $X$ $\pi$-phase-shifts respectively.  In particular, the $X$ $\pi$-phase-shift is the not gate.
 \end{definition}
Because the $\pi$-phases are given by \ref{ZX.pi:S3}, by the commutative spider theorem, it is immediate that they are phase shifts:
\begin{description}

						\item[{[PH]}]
						\namedlabel{ZX.pi:PH}{\bf [PH]}
						\hfil
						$
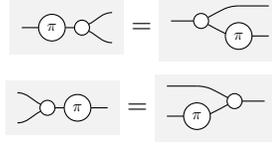
\begin{tikzpicture}
	\begin{pgfonlayer}{nodelayer}
		\node [style=xmul] (0) at (8, -0) {$\pi$};
		\node [style=xmul] (1) at (9, -0) {};
		\node [style=none] (2) at (7, -0) {};
		\node [style=none] (3) at (10, 0.5) {};
		\node [style=none] (4) at (10, -0.5) {};
	\end{pgfonlayer}
	\begin{pgfonlayer}{edgelayer}
		\draw [style=simple, in=-27, out=180, looseness=1.00] (4.center) to (1);
		\draw [style=simple] (1) to (0);
		\draw [style=simple] (0) to (2.center);
		\draw [style=simple, in=180, out=27, looseness=1.00] (1) to (3.center);
	\end{pgfonlayer}
\end{tikzpicture}
=
\begin{tikzpicture}
	\begin{pgfonlayer}{nodelayer}
		\node [style=xmul] (0) at (9, 0) {};
		\node [style=none] (1) at (8, 0) {};
		\node [style=none] (2) at (10.25, -0.5) {};
		\node [style=none] (3) at (10.25, 0.5) {};
		\node [style=xmul] (4) at (10.25, -0.5) {$\pi$};
		\node [style=none] (5) at (11.25, -0.5) {};
		\node [style=none] (6) at (11.25, 0.5) {};
	\end{pgfonlayer}
	\begin{pgfonlayer}{edgelayer}
		\draw [style=simple, in=27, out=180, looseness=1.00] (3.center) to (0);
		\draw [style=simple, in=180, out=-27, looseness=1.00] (0) to (2.center);
		\draw [style=simple] (6.center) to (3.center);
		\draw [style=simple] (2.center) to (5.center);
		\draw [style=simple] (0) to (1.center);
	\end{pgfonlayer}
\end{tikzpicture}
						$
						
						\hfil
						$
\begin{tikzpicture}
	\begin{pgfonlayer}{nodelayer}
		\node [style=xmul] (0) at (9, -0) {$\pi$};
		\node [style=xmul] (1) at (8, -0) {};
		\node [style=none] (2) at (10, -0) {};
		\node [style=none] (3) at (7, 0.5) {};
		\node [style=none] (4) at (7, -0.5) {};
	\end{pgfonlayer}
	\begin{pgfonlayer}{edgelayer}
		\draw [style=simple, in=-153, out=0, looseness=1.00] (4.center) to (1);
		\draw [style=simple] (1) to (0);
		\draw [style=simple] (0) to (2.center);
		\draw [style=simple, in=0, out=153, looseness=1.00] (1) to (3.center);
	\end{pgfonlayer}
\end{tikzpicture}
=
\begin{tikzpicture}
	\begin{pgfonlayer}{nodelayer}
		\node [style=xmul] (0) at (10.25, 0) {};
		\node [style=none] (1) at (11.25, 0) {};
		\node [style=none] (2) at (9, -0.5) {};
		\node [style=none] (3) at (9, 0.5) {};
		\node [style=xmul] (4) at (9, -0.5) {$\pi$};
		\node [style=none] (5) at (8, -0.5) {};
		\node [style=none] (6) at (8, 0.5) {};
	\end{pgfonlayer}
	\begin{pgfonlayer}{edgelayer}
		\draw [style=simple, in=153, out=0, looseness=1.00] (3.center) to (0);
		\draw [style=simple, in=0, out=-153, looseness=1.00] (0) to (2.center);
		\draw [style=simple] (6.center) to (3.center);
		\draw [style=simple] (2.center) to (5.center);
		\draw [style=simple] (0) to (1.center);
	\end{pgfonlayer}
\end{tikzpicture}
						$
\end{description}

In bra-ket notation, a black spider from $n$ to $m$ with angle $\theta$ is interpreted as follows in $\FdHilb$:
$$|0\>^{ \otimes n}\< 0|^{ \otimes m}+e^{ i \theta}|1\>^{ \otimes n}\< 1|^{ \otimes m}$$
and a white spider from $n$ to $m$ with angle $\theta$ is interpreted as follows in $\FdHilb$:
$$|+\>^{ \otimes n}\< +|^{ \otimes m}+e^{ i \theta}|-\>^{ \otimes n}\< -|^{ \otimes m}$$


Note that the controlled-not gate has a succinct representation in $\ZX_\pi$ (this can be verified by calculation):

$$
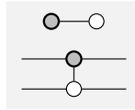
\begin{tikzpicture}
	\begin{pgfonlayer}{nodelayer}
		\node [style=rn] (0) at (-2.25, 2) {};
		\node [style=rn] (1) at (-2.25, 3) {};
		\node [style=rn] (2) at (1.25, 3) {};
		\node [style=rn] (3) at (1.25, 2) {};
		\node [style=zmul] (4) at (-0.5, 3) {};
		\node [style=xmul] (5) at (-0.5, 2) {};
		\node [style=zmul] (6) at (-1.25, 4.25) {};
		\node [style=xmul] (7) at (0.25, 4.25) {};
	\end{pgfonlayer}
	\begin{pgfonlayer}{edgelayer}
		\draw [style=simple] (5) to (4);
		\draw [style=simple] (4) to (1);
		\draw [style=simple] (4) to (2);
		\draw [style=simple, in=0, out=180, looseness=1.00] (3) to (5);
		\draw [style=simple] (5) to (0);
		\draw [style=simple] (7) to (6);
	\end{pgfonlayer}
\end{tikzpicture}
$$

This means that $\ZX_\pi$ contains all of the generators of the real Clifford group. Furthermore, the following is known:
\begin{theorem}\cite{realzx} \label{thm:zxcomplete} 
$\ZX_\pi$ is complete for real stabilizer states.
\end{theorem}

The original presentation of $\ZX_\pi$ in \cite{realzx} did not account for scalars; instead, it imposed the equivalence relation on circuits up to an invertible scalar and ignored the zero scalar entirely.  Therefore, the original completeness result described in \cite{realzx} is not actually as strong as Theorem \ref{thm:zxcomplete}.  This means, of course, that this original calculus does not embed in $\Mat_\C$ as the relations are not sound.   For example, the following map is interpreted as $\sqrt 2$, not $1$, in $\Mat_\C$:
$$
\begin{tikzpicture}
	\begin{pgfonlayer}{nodelayer}
		\node [style=xmul] (0) at (3, 2) {};
		\node [style=zmul] (1) at (4, 2) {};
	\end{pgfonlayer}
	\begin{pgfonlayer}{edgelayer}
		\draw [style=simple] (1) to (0);
	\end{pgfonlayer}
\end{tikzpicture}
$$

 Later on, \cite{scalars,removingstar}  showed that by scaling certain axioms to make them sound, and by adding Axioms \ref{ZX.pi:IV} and \ref{ZX.pi:ZO} this fragment of the ZX-calculus is also complete for scalars.   The properly scaled axioms have all been collected in Figure \ref{fig:ZXPI}.

\section{Embedding \texorpdfstring{$\CNOT$}{CNOT} into \texorpdfstring{$\ZX_\pi$}{the real stabilizer fragment of the ZX-calculus} }

Consider the interpretation of $\CNOT$ into $\ZX_\pi$, sending:

$$
\begin{tikzpicture}
	\begin{pgfonlayer}{nodelayer}
		\node [style=rn] (0) at (0, -0) {};
		\node [style=rn] (1) at (2, -0) {};
		\node [style=oplus] (2) at (1, -1) {};
		\node [style=dot] (3) at (1, -0) {};
		\node [style=rn] (4) at (2, -1) {};
		\node [style=rn] (5) at (0, -1) {};
	\end{pgfonlayer}
	\begin{pgfonlayer}{edgelayer}
		\draw [style=simple] (2) to (3);
		\draw [style=simple] (0) to (3);
		\draw [style=simple] (3) to (1);
		\draw [style=simple] (4) to (2);
		\draw [style=simple] (2) to (5);
	\end{pgfonlayer}
\end{tikzpicture}
\mapsto
\begin{tikzpicture}
	\begin{pgfonlayer}{nodelayer}
		\node [style=rn] (0) at (-2, 2) {};
		\node [style=rn] (1) at (-2, 3) {};
		\node [style=rn] (2) at (1, 3) {};
		\node [style=rn] (3) at (1, 2) {};
		\node [style=zmul] (4) at (-0.5, 3) {};
		\node [style=xmul] (5) at (-0.5, 2) {};
		\node [style=zmul] (6) at (-1.25, 4.25) {};
		\node [style=xmul] (7) at (0.25, 4.25) {};
	\end{pgfonlayer}
	\begin{pgfonlayer}{edgelayer}
		\draw [style=simple] (5) to (4);
		\draw [style=simple] (4) to (1);
		\draw [style=simple] (4) to (2);
		\draw [style=simple, in=0, out=180, looseness=1.00] (3) to (5);
		\draw [style=simple] (5) to (0);
		\draw [style=simple] (7) to (6);
	\end{pgfonlayer}
\end{tikzpicture}
\hspace*{.75cm}
\begin{tikzpicture}
	\begin{pgfonlayer}{nodelayer}
		\node [style=onein] (0) at (0, 0) {};
		\node [style=none] (1) at (1, 0) {};
	\end{pgfonlayer}
	\begin{pgfonlayer}{edgelayer}
		\draw (1.center) to (0);
	\end{pgfonlayer}
\end{tikzpicture}
\mapsto
\begin{tikzpicture}
	\begin{pgfonlayer}{nodelayer}
		\node [style=xmul] (0) at (1, -1) {$\pi$};
		\node [style=rn] (1) at (2, -1) {};
		\node [style=zmul] (2) at (1, -0) {};
		\node [style=xmul] (3) at (2, -0) {};
	\end{pgfonlayer}
	\begin{pgfonlayer}{edgelayer}
		\draw [style=simple] (1) to (0);
		\draw [bend left=45, looseness=1.25] (3) to (2);
		\draw [bend left=45, looseness=1.25] (2) to (3);
		\draw (3) to (2);
	\end{pgfonlayer}
\end{tikzpicture}
\hspace*{.75cm}
\begin{tikzpicture}[xscale=-1]
	\begin{pgfonlayer}{nodelayer}
		\node [style=onein] (0) at (0, 0) {};
		\node [style=none] (1) at (1, 0) {};
	\end{pgfonlayer}
	\begin{pgfonlayer}{edgelayer}
		\draw (1.center) to (0);
	\end{pgfonlayer}
\end{tikzpicture}
\mapsto
\begin{tikzpicture}
	\begin{pgfonlayer}{nodelayer}
		\node [style=xmul] (0) at (2, -1) {$\pi$};
		\node [style=rn] (1) at (1, -1) {};
		\node [style=zmul] (2) at (1, -0) {};
		\node [style=xmul] (3) at (2, -0) {};
	\end{pgfonlayer}
	\begin{pgfonlayer}{edgelayer}
		\draw [style=simple] (1) to (0);
		\draw [bend left=45, looseness=1.25] (3) to (2);
		\draw [bend left=45, looseness=1.25] (2) to (3);
		\draw (3) to (2);
	\end{pgfonlayer}
\end{tikzpicture}
$$



We explicitly prove that this interpretation is functorial.

\begin{lemma}\label{lem:cnotzxfunc}
The interpretation of $\CNOT$ into $\ZX_\pi$ is functorial.
\end{lemma}
\begin{proof}
See \ref{sec:embeddingcnotzx} \ref{lem:cnotzxfunc:proof}
\end{proof}

Because the standard interpretations of $\CNOT$ and $\ZX_\pi$ into $\Mat_\C$ commute and are faithful, the following diagram of strict \dag-symmetric monoidal functors makes $\CNOT\to\ZX_\pi$ faithful: 
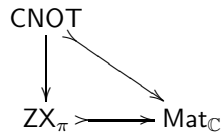
$$
\xymatrix{
\CNOT \ar@{ >->}[dr] \ar[d]\\
\ZX_\pi \ar@{ >->}[r] & \Mat_\C
}
$$

\section{Extending \texorpdfstring{$\CNOT$}{CNOT} to \texorpdfstring{$\ZX_\pi$}{the real stabilizer fragment of the ZX-calculus}}

As opposed to the ZX-calculus,  the identities  of $\CNOT$ are given in terms of {\em circuit relations}.  When applying rules of the $\ZX$ calculus, circuits can be transformed into intermediary representations so that the flow of information is lost.  Various authors have found complete circuit relations for various fragments of quantum computing.  Notably, Selinger found a complete set of identities for Clifford circuits (stabilizer circuits without ancillary bits) \cite{selinger2015generators}.  Similarly, Amy et al. found a complete set of identities for cnot-dihedral circuits (without ancillary bits) \cite{1701.00140}.

In this section, we provide a complete set of circuit relations for {\em real} stabilizer circuits (although circuits can have norms greater than 1).
We show that $\CNOT$ is embedded in $\ZX_\pi$ and we complete $\CNOT$ to $\ZX_\pi$ by adding the Hadamard gate and the scalar $\sqrt 2$ as generators along with 5 relations.   

\begin{definition}
\label{def:cnoth}

Let $\CNOT+H$ denote the PROP freely generated by the axioms of $\CNOT$ with additional generators the Hadamard gate and $\sqrt{2}$:

$$

$$
 
 satisfying the following identities:
\begin{figure}[H]
	\noindent
	\scalebox{1.0}{%
		\vbox{%
			\begin{mdframed}
				\begin{multicols}{2}
\begin{description}
\item[{[H.I]}]
\namedlabel{cnoth:H.I}{\bf [H.I]}
\label{H.I}
\hfil$
%
$

\end{description}
\end{multicols}
				\
			\end{mdframed}

	}}
	\caption{The identities of $\CNOT+H$ (in addition to the identities of $\CNOT$)}
	\label{fig:CNOTH}
\end{figure}
\end{definition}

The inverse of $\sqrt{2}$ is given the alias:
$$
%
$$

\ref{cnoth:H.I} is stating that the Hadamard gate is self-inverse.
\ref{cnoth:H.F} reflects the fact that composing the controlled-not gate with Hadamards reverses the control and operating bits.
\ref{cnoth:H.Z'} is stating that $\sqrt{2}$ composed with the zero matrix is again the zero matrix.
\ref{cnoth:H.S} makes  $\sqrt{2}$ and $1/\sqrt{2}$ inverses to each other.

\ref{cnoth:H.L'}  can be restated to resemble \ref{ZX.pi:S1}: 

\begin{lemma} $ $

\begin{description}
\item[{[H.L']}]
\namedlabel{cnoth:H.L}{\bf [H.L']}

\hfil$
%
 & \text{\ref{cnoth:H.L}}\\
\end{align*}
\end{proof}

\ref{cnoth:H.Z'} can be restated in slightly different terms:

\begin{lemma} $ $
 
\begin{description}
\item[{[H.Z']}]
\namedlabel{cnoth:H.Z}{\bf [H.Z']}
\hfil
$
%
$
\end{description}
\end{lemma}
\begin{proof}
Immediate by \ref{cnoth:H.S} and \ref{cnoth:H.Z'}.
\end{proof}

There is a derived identity, showing that the Frobenius structure identified with the inverse products of $\CNOT$ is unital:

\begin{lemma} $ $

\begin{description}
\item[{[H.U]}]
\namedlabel{cnoth:H.U}{\bf [H.U]}

\hfil
$
%
& \ref{cnoth:H.S}
\end{align*}

\end{proof}

\subsection{The completeness of \texorpdfstring{$\CNOT+H$}{CNOT+H}}

We construct two functors between $\CNOT+H$ and $\ZX_\pi$ and show that they are pairwise inverses.

\begin{definition}
Let  $F:\CNOT+H \to \ZX_\pi$, be the extension of the interpretation $\CNOT \to \ZX_\pi$ which takes:

$$
%
$$

Let  $G:\ZX_\pi\to \CNOT+H $ be the interpretation sending:
$$
%
$$
\end{definition}

The scalars $\sqrt{2}^n$ are taken to their chosen representatives by $F$ and $G$:

\begin{lemma}\
\label{lem:cantthink}
\begin{enumerate}[label=(\roman*)]
\item
$
\hfil
%
$

\end{enumerate}
\end{lemma}

\begin{proof}
See Lemma \ref{lem:cantthink:proof}
\end{proof}

This lemma makes it easier to show that $F$ and $G$ are functors:
\begin{lemma}
\label{lem:fisfunc}
$F:\CNOT+H \to \ZX_\pi$ is a strict \dag-symmetric monoidal functor.
\end{lemma}
\begin{proof}
See \ref{sec:completeness} Lemma \ref{lem:fisfunc:proof} 
\end{proof}

For the other way around:

\begin{lemma}
\label{lem:gisfunc}
$G:\ZX_\pi\to\CNOT+H $ is a strict \dag-symmetric monoidal functor.
\end{lemma}
\begin{proof}
See Lemma \ref{lem:gisfunc:proof}
\end{proof}

Next:

\begin{proposition}
\label{prop:fginv}
$\CNOT+H\xrightarrow{F} \ZX_\pi$ and $\ZX_\pi\xrightarrow{G} \CNOT+H$ are inverses.
\end{proposition}

\begin{proof}
See Proposition \ref{prop:fginv:proof}
\end{proof}

Because all of the axioms of $\CNOT+H$ and $\ZX_\pi$ satisfy the same ``horizontal symmetry''; we can not only conclude that they are isomorphic, but rather:

\begin{theorem}
$\CNOT+H$ and $\ZX_\pi$ are strictly $\dag$-symmetric monoidally isomorphic.
\end{theorem}

\section{Towards the Toffoli gate plus the Hadamard gate}
\label{section:tof}

Recall the PROP $\TOF$, generated by the 1 ancillary bits $\onein$ and $\oneout$ (depicted graphically as in $\CNOT$) as well as the Toffoli gate:   
  \[  \tof := 
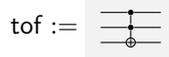
		\begin{tikzpicture}
		\begin{pgfonlayer}{nodelayer}
		\node [style=nothing] (0) at (0, -0) {};
		\node [style=nothing] (1) at (0, 0.5000001) {};
		\node [style=nothing] (2) at (0, 1) {};
		\node [style=nothing] (3) at (2, 1) {};
		\node [style=nothing] (4) at (2, 0.5000001) {};
		\node [style=nothing] (5) at (2, -0) {};
		\node [style=dot] (6) at (1, 1) {};
		\node [style=dot] (7) at (1, 0.5000001) {};
		\node [style=oplus] (8) at (1, -0) {};
		\end{pgfonlayer}
		\begin{pgfonlayer}{edgelayer}
		\draw (2) to (6);
		\draw (6) to (3);
		\draw (4) to (7);
		\draw (7) to (1);
		\draw (0) to (8);
		\draw (8) to (5);
		\draw (8) to (7);
		\draw (6) to (7);
		\end{pgfonlayer}
		\end{tikzpicture}
		\]

The axioms are given Figure \ref{fig:TOF}, which we have put in the Appendix \ref{sec:tof}.
Recall that have that:

\begin{theorem}
$\TOF$ is discrete-inverse equivalent to the category of partial isomorphisms between finite powers of the two element set, and thus, is complete.
\end{theorem}

By \cite{aaronson}, we have that the Toffoli gate is universal for classical reversible computing, therefore $\TOF$ is a complete set of identities for the universal fragment of classical computing.
However, the category is clearly is not universal for quantum computing.  Surprisingly, by adding the Hadamard gate as a generator, this yields a category which is universal for an approximately universal fragment of quantum computing \cite{tofh}.

Thus, one would hope that the completeness of $\CNOT+H$ could be used to give a complete set of identities for a category $\TOF+H$.

Although we have not found such a complete set of identities,  the identity \ref{cnoth:H.F} can be easily extended to an identity that characterizes the commutativity of a multiply controlled-$Z$-gate.  This could possibly facilitate a two way translation to and from the ZH calculus \cite{zh}, like we performed between $\CNOT+H$ and $\ZX_\pi$.  This, foreseeably would be much easier than a translation between one of the universal fragments of the ZX-calculus; because, despite the recent simplifications of the Toffoli gate in terms of the triangle, the triangle itself does not have a simple representation in terms of the Toffoli gate, Hadamard gate and computational ancillary bits \cite{vilmart}.

If we conjugate the not gate ($X$ gate) with Hadamard gates, we get the $Z$ gate:

$$

$$

  Furthermore, if we conjugate the operating bit controlled-not gate with the the Hadamard gate, we get the controlled-$Z$ gate:
$$
%
$$  
  
Because the flow of information in the  controlled-$Z$ gate is undirected in the sense that:
$$
%
$$

this motivates the identity $\ref{cnoth:H.F}$ of $\CNOT+H$:

$$
%
$$

We can continue this, so that by conjugating the operating bit of the Toffoli gate with Hadamard gates, we obtain a doubly controlled-$Z$ gate.  This suggests the following (sound) identity:
\begin{description}
\item[{[H.F']}]
\namedlabel{tofh:H.F}{\bf [H.F']}
\hfil
$
%
$$

So that we can unambiguously represent the doubly controlled-$Z$ gate as:
$$
%
$$

Indeed, this identity entails a more general form for generalized controlled-not gate with two or more controls.  Recall the definition of a multiply controlled-not gate in \cite{tof}:

\begin{definition}\cite[Definition 5.1]{tof}
\label{definition:Generalizedcnot}
For every $n\in\N$, inductively define the controlled not gate, $\cnot_n:n+1\to n+1$ inductively by:
\begin{itemize}
	\item
	For the base cases, let $\cnot_0:=\notgate$, $\cnot_1:=\cnot$ and $\cnot_2:=\tof$.

	\item
	For all $n \in \N$ such that $n\geq 2$:
	
	\[
	\cnot_{n+1} \equiv
%
	\]	
	
\end{itemize}
\end{definition}


Recall $\cnot_n$ gates can be decomposed into other $\cnot_n$ gates in the following fashion:

\begin{proposition}\cite[Proposition 5.3 (i)]{tof}
$\cnot_{n+k}$ gates can be zipped and unzipped:
\[
%
\]
\end{proposition}

Therefore, can can derive that:

\begin{lemma}
\label{lem:h.f}

\ref{tofh:H.F} entails:
$$
%
$$
\end{lemma}
\begin{proof}
From the zipper lemma and \ref{tofh:H.F}, we have:

$$
%
$$
\end{proof}

Recall from \cite[Corollary 5.4]{tof} that control-wires of $\cnot_n$ gates can be permuted in the following sense:
	$$
	%
	$$
Therefore, by this observation and Lemma \ref{lem:h.f} we have:
$$
%
$$

So that, by observing that the multiply controlled-$Z$ gate can be unambiguously defined as follows:
$$
%
$$

\section*{Acknowledgement}
The author would like to thank Robin Cockett, Jean-Simon Lemay and John van de Wetering for useful discussions.

  \nocite{duncaninteracting}
  \nocite{bonchiinteracting}
  \nocite{gottesmanstabilizer}
  
  \bibliography{cnoth}
  \bibliographystyle{eptcs}

\appendix

\section{Embedding \texorpdfstring{$\CNOT$}{CNOT} into \texorpdfstring{$\ZX_\pi$}{the real stabilizer fragment of the ZX-calculus} (proofs)}
\label{sec:embeddingcnotzx} 

The following basic identity is needed:

\begin{lemma} \cite[Lemma 19]{vilmart}
\label{lem:dim}
\label{lem:dim:proof}

\end{align*}
\end{proof}

We also need:
\begin{lemma}
\label{lem:cnottozx}
\begin{enumerate}[label=(\roman*)]\
\item
\hfil
$
%
\end{align*}
\endgroup
\end{enumerate}
\end{proof}

\begin{customlemma}{\ref{lem:cnotzxfunc}}
\label{lem:cnotzxfunc:proof}
The interpretation of $\CNOT$ into $\ZX_\pi$ is functorial.
\end{customlemma}

\begin{proof}
We prove that each axiom holds

\begin{description}
\item[\ref{CNOT.2}]
\begin{align*}
%
\end{align*}

\item[\ref{CNOT.3}]
Immediate from commutative spider theorem.

\item[\ref{CNOT.4}]
\begin{align*}
%
& \text{Lemma \ref{lem:cnottozx} (ii)}
\end{align*}

\item[\ref{CNOT.5}]
Immediate from commutative spider theorem.

\item[\ref{CNOT.6}]
Frobenius algebra is special.

\item[\ref{CNOT.7}]
\begin{align*}
%
 & \text{Lemma \ref{lem:cnottozx} (ii)}
\end{align*}
\endgroup

\end{description}
\end{proof}

\subsection{Proof of Lemma \ref{lem:cantthink}}
\label{sec:completeness}

\label{lem:cantthink:proof}

\begin{customlemma}{\ref{lem:cantthink} }\

\begin{enumerate}[label=(\roman*)]
\item
$
\hfil
%
 & \ref{cnoth:H.S}\\
\end{align*}
\end{enumerate}
\end{proof}

\subsection{Proof of Lemma \ref{lem:fisfunc}}

\label{lem:fisfunc:proof}

We show that these interpretations are functors:
\begin{customlemma}{\ref{lem:fisfunc}}

$F:\CNOT+H \to \ZX_\pi$ is a strict \dag-symmetric monoidal functor.
\end{customlemma}
\begin{proof}
The preservation of the \dag-symmetric monoidal structure is immediate.
As the restriction of $F$ to $\CNOT$ is a functor, it suffices to show that \ref{cnoth:H.I}, \ref{cnoth:H.F}, \ref{cnoth:H.U},  \ref{cnoth:H.L}, \ref{cnoth:H.S} and \ref{cnoth:H.Z} hold.

\begin{description}
\item[\ref{cnoth:H.I}] Immediate.
\item[\ref{cnoth:H.F}]
\begin{align*}
 & \text{Lemma \ref{lem:cnottozx} (ii)}
\end{align*}
\endgroup

\end{description}

\end{proof}

\subsection{Proof of Lemma \ref{lem:gisfunc} }
\label{lem:gisfunc:proof}

\begin{customlemma}{\ref{lem:gisfunc} }\
$G:\ZX_\pi\to\CNOT+H $ is a strict \dag-symmetric monoidal functor.
\end{customlemma}

\begin{proof}

We prove that each axiom holds:
\begin{description}

\item[\ref{ZX.pi:PI}]
  This follows by naturality of $\Delta$ in $\CNOT$.

\item[\ref{ZX.pi:B.U}]
  This follows by naturality of $\Delta$ in $\CNOT$ and \ref{cnoth:H.S}.

\item[\ref{ZX.pi:H2}]
  This follows immediately from \ref{cnoth:H.I}.
  
  \item[\ref{ZX.pi:H2}]
  This follows immediately from \ref{cnoth:H.S}.

\item[\ref{ZX.pi:PP}]
\begin{align*}

\end{align*}
\endgroup

\item[Classical structure: ]  Remark that  rules  \ref{cnoth:H.U} and  \ref{cnoth:H.S} complete the semi-Frobenius structure to the appropriate classical structure.
\end{description}
\end{proof}

\subsection{Proof of Proposition \ref{prop:fginv}}
\label{prop:fginv:proof}

\begin{customproposition}{\ref{prop:fginv}}\
$\CNOT+H\xrightarrow{F} \ZX_\pi$ and $\ZX_\pi\xrightarrow{G} \CNOT+H$ are inverses.
\end{customproposition}

\begin{proof}\

\begin{enumerate}
\item First, we show that $G;F=1$ .
  
  We only prove the cases for the generators $\cnot$ and $\onein$ as the claim follows trivially for the Hadamard gate and by symmetry for $\oneout$:
  
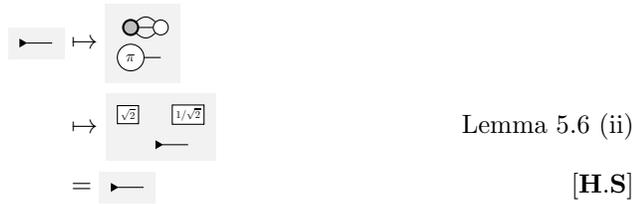
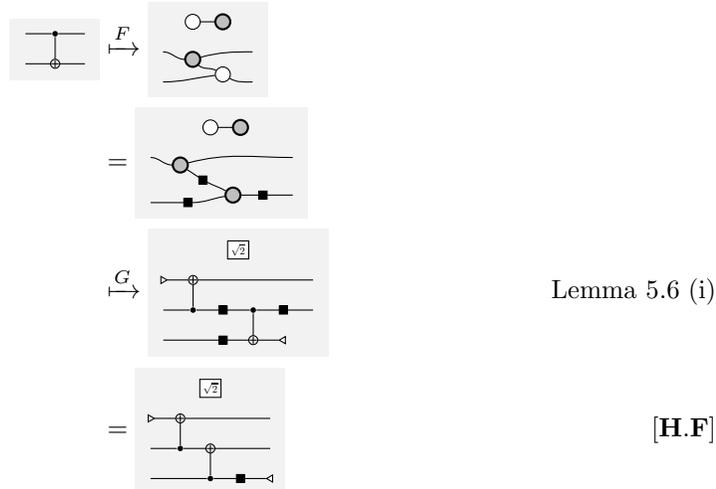
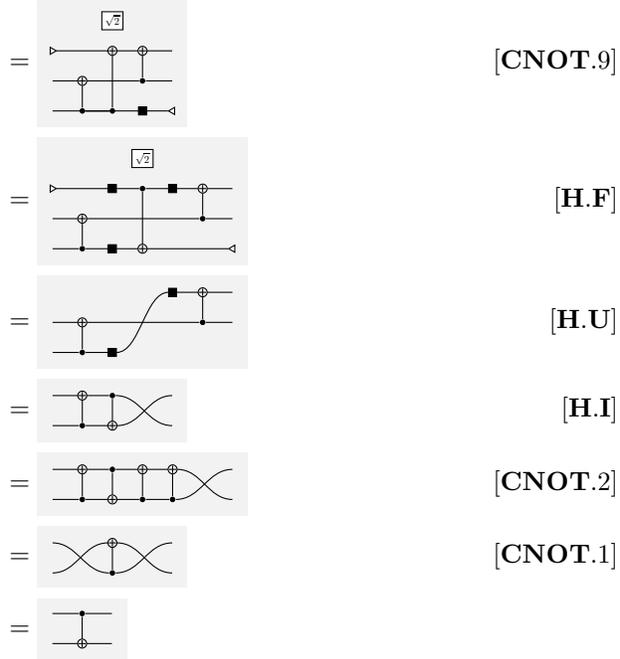
\begin{description}
\item[For $\onein$:]
\begin{align*}

\end{align*}
\endgroup

\item It is trivial to observe, on the other hand, that $G;F=1$.
\end{description}

\end{enumerate}
\end{proof}

\section{The identities of $\TOF$}

\label{sec:tof}

$\TOF$ is the PROP generated by the 1 ancillary bits $\onein$ and $\oneout$ as well as the Toffoli gate:   
  \[ 
   \onein := 
   %
		\]
		
		Where there is the derived generator
		$$
		   \cnot = 
%
$$
		and the not gate and $|0\>$ ancillary bits are dervied as in Section \ref{sec:cnot}.  These generators must satisfy the identities given in the following figure:
		
\begin{figure}[H]
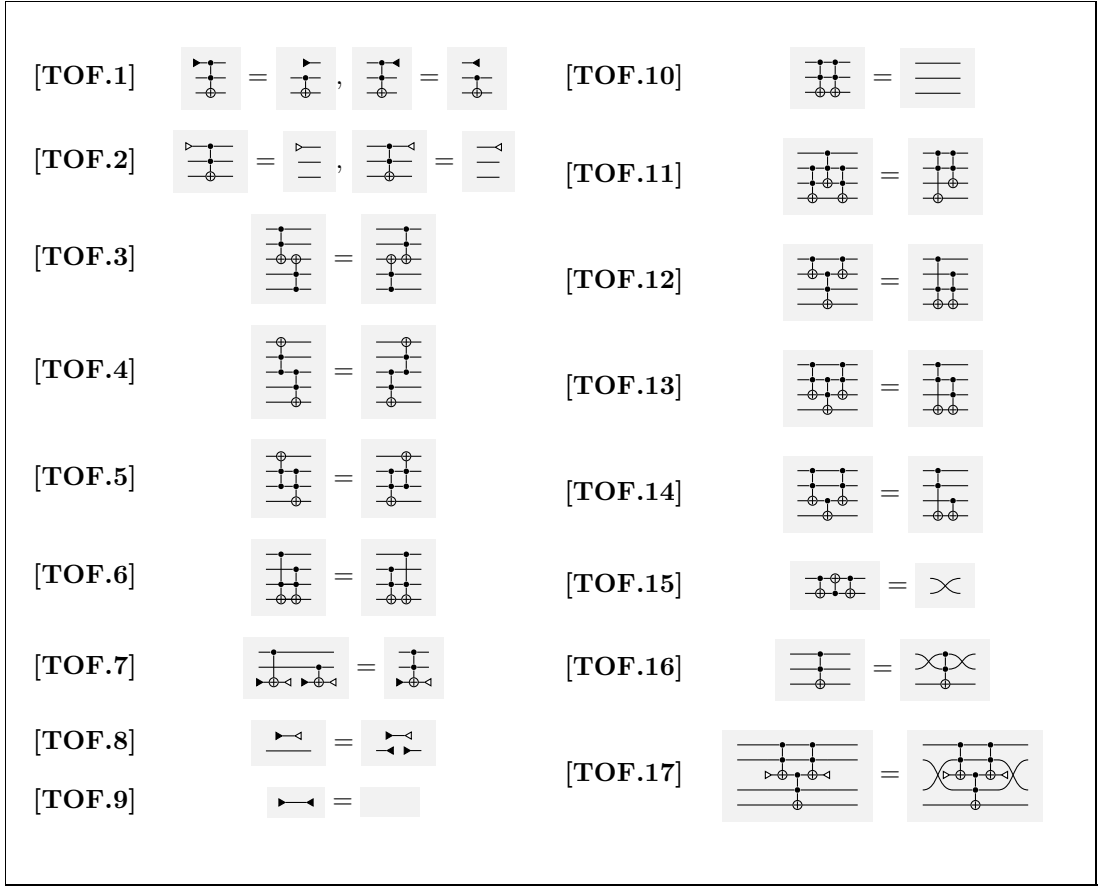

\label{fig:tof}

\noindent
\scalebox{1.0}{%
\vbox{%
\begin{mdframed}
\begin{multicols}{2}
\begin{enumerate}[label={\bf [TOF.\arabic*]}, ref={\bf [TOF.\arabic*]}, wide = 0pt, leftmargin = 2em]
\item
\label{TOF.1}
{\hfil
$
%
$}
\end{enumerate}
\end{multicols}
\
\end{mdframed}
}}
\caption{The identities of $\TOF$}
\label{fig:TOF}
\end{figure}

\section*{Errata}
The author would like to apologize for giving an incorrect proof of Corollary 4.9 in the previous preprint version; meaning that the scalar $\sqrt 2$ cannot be removed.

\end{document}